\documentclass[11pt]{amsart}
\usepackage[a4paper, total={6in, 8in}]{geometry}
\usepackage{enumerate}
\usepackage{amsmath}
\usepackage{amssymb,latexsym}
\usepackage{amsthm}
\usepackage{color}
\usepackage{cancel}
\usepackage{graphicx}
\usepackage{cite}
\usepackage{changes}
\usepackage{hyperref}
\usepackage{comment}
\usepackage{amscd}
\usepackage{mathtools}

\DeclareMathOperator{\C}{\mathcal{C}}

\usepackage[foot]{amsaddr}

\DeclareMathOperator{\Aut}{Aut}
\DeclareMathOperator{\End}{End}
\DeclareMathOperator{\Gal}{Gal}
\DeclareMathOperator{\supp}{supp}
\DeclareMathOperator{\rk}{rk}

\theoremstyle{definition}
\newtheorem{theorem}{Theorem}[section]

\newtheorem{lemma}[theorem]{Lemma}
\newtheorem{corollary}[theorem]{Corollary}

\newtheorem{proposition}[theorem]{Proposition}

\newtheorem{remark}[theorem]{Remark}

\newcommand{\fqn}{\mathbb{F}_{q^n}}

\newcommand{\cC}{{\mathcal C}}
\newcommand{\cS}{{\mathcal S}}

\newcommand{\cT}{{\mathcal T}}

\newcommand{\cG}{{\mathcal G}}

\newcommand{\cH}{{\mathcal H}}
\newcommand{\cD}{{\mathcal D}}
\newcommand{\cV}{{\mathcal V}}

\newcommand{\cL}{{\mathcal L}}

\newcommand{\F}{{\mathbb F}}

\newcommand{\K}{{\mathbb K}}
\newcommand{\LL}{{\mathbb L}}
\newcommand{\fq}{{\mathbb F}_{q}}

\newcommand{\la}{\langle}
\newcommand{\ra}{\rangle}

\newcommand{\PG}{\mathrm{PG}}
\newcommand{\N}{\mathrm{N}}

\newcommand{\ZZ}[1]{\mathbb Z/#1\mathbb Z}

\newcommand{\alessandro}[1]{{\color{blue} \sf $\star\star$ Alessandro: [#1]}}

\title{Extending two families of maximum rank distance codes}
\date{}

\author[A. Neri]{Alessandro Neri}
\address{Alessandro Neri, \textnormal{Max-Planck-Institute for Mathematics in the Sciences, Inselstraße 22, 04103 Leipzig, Germany}}
\email{alessandro.neri@mis.mpg.de}

\author[P. Santonastaso]{Paolo Santonastaso}
\address{Paolo Santonastaso, \textnormal{Dipartimento di Matematica e Fisica, Universit\`a degli Studi della Campania ``Luigi Vanvitelli'', Viale Lincoln, 5, I--\,81100 Caserta, Italy}}
\email{paolo.santonastaso@unicampania.it}

\author[F. Zullo]{Ferdinando Zullo}
\address{Ferdinando Zullo, \textnormal{Dipartimento di Matematica e Fisica, Universit\`a degli Studi della Campania ``Luigi Vanvitelli'', Viale Lincoln, 5, I--\,81100 Caserta, Italy}}
\email{ferdinando.zullo@unicampania.it}

\subjclass[2020]{11T71; 11T06; 94B05} 
\keywords{Rank-metric codes; linearized polynomials;  MRD codes; scattered polynomials}
\begin{document}

\maketitle

\begin{abstract}
In this paper we  provide a large family of rank-metric codes, which contains properly the codes recently found by Longobardi and Zanella (2021) and by Longobardi, Marino, Trombetti and Zhou (2021). These codes are $\F_{q^{2t}}$-linear of dimension $2$ in the space of linearized polynomials over $\F_{q^{2t}}$, where $t$ is any integer greater than $2$, and we prove that they are maximum rank distance codes.
For $t\ge 5$, we determine their equivalence classes and these codes turn out to be inequivalent to any other construction known so far, and hence they are really new.
\end{abstract}

\section{Introduction}

Codes endowed with the rank-metric have gained a lot of interest in the last decade due to their numerous applications. In particular, the turning point was the groundbreaking work of Silva, K\"otter and Kschischang  \cite{silva2008rank}, in which they proposed rank metric codes as tools for dealing with linear random network coding. However, the origin of rank-metric codes is due to Delsarte's seminal paper  \cite{de78} in  1978, where they were first defined for a pure combinatorial interest. Few years later, Gabidulin rediscovered them independently \cite{ga85a}. The first applications were due to Roth in \cite{roth1991maximum} for crisscross deletion correction, and to Gabidulin, Paramonov and Tretjakov in \cite{gabidulin1991ideals} for a  cryptosystem based on rank-metric codes. From a mathematical point of view, rank-metric codes have been shown to possess connections with many subjects, such as semifield theory \cite{sheekey2016new}, linear sets in finite geometry \cite{polverino2020connections}, tensorial algebra \cite{byrne2019tensor},  skew algebras \cite{ACLN20+,elmaazouz2021}, matroid theory \cite{gorla2020rank} and many more.  All these connections testify the rich structure that rank-metric code possess. 

Formally, rank-metric codes are sets of $n \times m$ matrices over a finite field $\F_q$, endowed with the rank metric. This is the metric defined by the rank distance, where the rank distance between two matrices is the rank of their difference. There is also another representation of rank-metric codes, which allows to endow them with stronger algebraic properties. When $n=m$, one may indeed identify the space $\fq^{n\times n}$, with the ring of $\sigma$-polynomials with coefficients in $\fqn$, where $\sigma$ is a generator of $\Gal(\fqn/\fq)$. This allows also to introduce a notion of $\fqn$-linearity of a rank-metric code.

Among rank-metric codes, of particular interest is the family of \emph{maximum rank distance (MRD)} codes. These are codes that have optimal parameters: for the given size  and minimum rank distance, they have the maximum cardinality. The first construction of a family of MRD codes was due already to Delsarte \cite{de78} and independently to Gabidulin \cite{ga85a}. The codes of this family are now known as \emph{Gabidulin codes} and for many years were essentially the only known constructions, until Sheekey came up with a broader family of MRD codes, named \emph{twisted Gabidulin codes}.
The flexibility of the parameters of these families is one of the main reasons that made Gabidulin and twisted Gabidulin codes very appealing. Another large family of MRD codes was later given by Trombetti and Zhou in \cite{trombetti2018new}. Except from that, almost all the other known MRD construction have all very specific restriction on some of the parameters: for instance, the MRD codes   \cite{BZZ,csajbok2018anewfamily,csajbok2018maximum,csajbok2018linearset,MMZ,ZZ}  exist only for $n \in \{6,7,8\}$.

Very recently, two classes of $2$-dimensional $\fqn$-linear MRD codes have been introduced for any even $n$. One was given in \cite{longobardizanellascatt} and the second in \cite{longobardi2021large}. The arguments used there for showing that these codes are MRD exploited the correspondence of MRD codes with scattered linear sets \cite{polverino2020connections} and in particular with scattered polynomials.  

\medskip

In this paper we provide a wider family of $2$-dimensional $\fqn$-linear MRD codes properly containing the two families introduced in \cite{longobardizanellascatt,longobardi2021large}; see Theorem \ref{th:newfamily}. In order to show that they are MRD, we first give in Theorem \ref{th:rankfs} a more general argument which allows to extend any construction of MRD codes based on $\sigma$-polynomials to any other generator $\theta$ of $\Gal(\fqn/\fq)$ under certain hypotheses. We then focus on the study of the equivalence of these codes. We first prove that the codes in this wider family  are all inequivalent to all the other $\fqn$-linear MRD codes known so far; see Proposition \ref{prop:not_Gabidulin} and Theorem \ref{th:no_twisted}. Afterwards, we concentrate on the equivalence problem  within this new family: in Theorem \ref{th:equivspecialcase} and in Corollary \ref{cor:inequivalence_newMRD}, we characterize for which parameters two of these codes are equivalent. This also allows to derive results on the number of equivalence classes of codes in the new family: in Theorem \ref{th:numexact} we provide the exact but implicit number of these equivalence classes, while in Theorem \ref{th:lowbound} we give an explicit lower bound.

The paper is structured as follows. Section \ref{sec:preliminaries} collects all the basic ingredients we need throughout the paper.
In Section \ref{sec:newfamily} we introduce the new family of codes, showing that they are MRD. Section \ref{sec:equivalence} is dedicated to the study of the equivalence classes of the new codes. Finally, we recap our findings and list some open problems in Section \ref{sec:conclusions}.

\section{Preliminaries} \label{sec:preliminaries}

In this section we give a recap on the important notions and results needed for the paper.  We start introducing rank-metric codes and  their representation as linearized polynomials. We then recall the notion of maximum rank distance (MRD) codes and explain the two most prominent infinite families of MRD codes known up to now. Finally, we give a short description of the invariants studied in \cite{neri2020equivalence}, which will be used for determining code inequivalence. For the interested reader, we refer to the  survey on rank-metric codes written by Sheekey \cite{sheekeysurvey}, which provides an exhaustive study on rank-metric codes in various frameworks.

We start fixing the following notation. Let $p$ be a prime and $r$ a positive integer. We fix $q=p^r$ and denote by $\fq$ the finite field with $q$ elements. Moreover, we fix a positive integer $t$, let $n=2t$, and consider the extension field $\fqn$ of degree $n$ over $\fq$. It is well-known that this extension is Galois and that the Galois group $\Gal(\fqn/\fq)$ is cyclic. For the rest of the paper we will use $\sigma$ and $\theta$ to denote  generators of $\Gal(\fqn/\fq)$. Recall that for a Galois extension of the form $\fqn/\fq$, the \emph{norm} of an element $\alpha \in \fqn$ is defined as
$$ \mathrm{N}_{q^n/q}(\alpha):= \prod_{\rho\in\Gal(\fqn/\fq)}\rho(\alpha).$$

\subsection{Rank-metric codes and linearized polynomials}\label{sec:rankmetric_linearized}

Rank-metric codes were introduced by Delsarte \cite{de78} in 1978 and they have been intensively investigated in recent years because of their applications in crisscross error correction \cite{roth1991maximum}, cryptography \cite{gabidulin1991ideals} and network coding \cite{silva2008rank}. 
Formally, on the set of  matrices $\F_q^{n\times m}$ we can define the \emph{rank-metric}, as
\[d(A,B) = \mathrm{rk}\,(A-B), \qquad \mbox{ for } A,B\in\fq^{n\times m}.\]
 A \emph{rank-metric code}  is a subset $\C$  of $\F_q^{n \times m}$ endowed with the rank metric.
The \emph{minimum rank distance} of $\C$ is defined as
\[d := d(\C) =\min\{ d(A,B) \colon A,B \in \C,\,\, A\neq B \}.\]
Moreover, if $\C$ is an  $\fq$-linear subspace of $\fq^{n\times m}$, we will also say that the code is $\fq$-linear, and in this case  the minimum rank distance is also equal to
\[d(\C)= \min\{ \rk(A) \colon A \in \C,\,\, A\neq 0 \}.\]
Delsarte showed in \cite{de78} that the parameters of a rank-metric code  must satisfy a Singleton-like bound, that reads as
\[ |\C| \leq q^{\max\{m,n\}(\min\{m,n\}-d+1)}. \]
When equality holds, we call $\C$ a \emph{maximum rank distance} (\emph{MRD} for short) code.
It was shown that  MRD codes exist for any choice of $q,n,m,d$; see \cite{de78,ga85a}.

In this paper we will focus on the case of square matrices, that is when $n=m$.\footnote{We remark that for the study of MRD codes, considering the case of square matrices is not a real restriction. Indeed if we have $m<n$, then any MRD code in $\fq^{n \times n}$ can be used to obtain MRD codes in $\fq^{n \times m}$ by simply removing from each matrix the last $n-m$ columns. 
} In this case there is an alternative way to see the $\fq$-algebra of $n\times n$ matrices as the algebra of  \emph{$\sigma$-polynomials}. Formally, let $\sigma$ be a generator of $\Gal(\fqn/\fq)$. A $\sigma$-polynomial is an element of the form
$$ f(x):=\sum_{i=0}^{n-1}f_i x^{\sigma^i}, \quad f_i \in \fqn.$$

The set of $\sigma$-polynomials forms a ring with the usual addition and the \emph{composition}, given by
$$ (f_ix^{\sigma^i})\circ (g_jx^{\sigma^j})=f_i\sigma^i(g_j)x^{\sigma^{i+j}}, $$
on $\sigma$-monomials, and then extended by distributivity. We denote this ring by $\cL_{n,\sigma}$. It is well-known that 
\begin{equation}\label{eq:isom_sigma} (\cL_{n,\sigma},+,\circ)\cong(\End_{\fq}(\fqn),+,\circ),
\end{equation}where the $\sigma$-polynomial $f(x)$ is identified with the endomorphism of $\fqn$
$$ \alpha \longmapsto \sum_{i=0}^{n-1}f_i \sigma^i(\alpha).$$
However, a different choice of the generator $\sigma$ only gives a different  representation of an element, but the ring $\cL_{n,\sigma}$ does not depend on this choice. Hence, $\cL_{n,\sigma}$ is the same for any choice of $\sigma$, and one can just consider  $\cL_{n,q}:=\cL_{n,\theta}$, where $\theta$ is the $q$-Frobenius isomorphism of $\fqn$, mapping each $\alpha \in \fqn$ to $\alpha^q$. Thus, one can speak of $\sigma$-polynomials just in $\cL_{n,q}$. Notice that $\cL_{n,q}$ is  isomorphic to the algebra of \emph{$q$-linearized polynomials} modulo the two-sided ideal generated by $x^{q^n}-x$; see \cite{wu2013linearized}. 

We remark that the isomorphism given in \eqref{eq:isom_sigma} holds in a more general setting and not only over finite fields. For instance, one can define $\sigma$-polynomials over any field $\LL$ with $\sigma \in \Aut(\LL)$, and obtain $(\cL_{n,\sigma},+,\circ)\cong(\End_{\K}(\LL),+,\circ)$, where $\K=\LL^\sigma$ and $n=[\LL:\K]$; see e.g. \cite{Gow}. More generally, there is a similar isomorphism also for \emph{any} Galois extension of fields  or of rings \cite[Theorem 1.3]{chase1969galois}, which have been exploited for the development of a more general theory of rank-metric codes \cite{augot2013rank,ACLN20+,kamche2019rank,elmaazouz2021}.

Thanks to the isomorphism in \eqref{eq:isom_sigma}, we immediately get that $(\cL_{n,\sigma},+,\circ)$ is also isomorphic to the $\fq$-algebra $\fq^{n\times n}$, since $\fqn$ is an $n$-dimensional $\fq$-vector space. Thus, rank-metric codes can equivalently be represented as subsets of $\cL_{n,q}$. Here, we will speak of \emph{kernel} and \emph{rank} of a $\sigma$-polynomial meaning by this the kernel and rank of the corresponding endomorphism. This naturally defines the rank-metric directly on $\cL_{n,q}$. 

There is a very special $\sigma$-polynomial that is central for many aspects of Galois theory and duality theories. This is the case of the \emph{trace map}, defined as
$$ \mathrm{Tr}_{q^n/q}(x):=\sum_{i=0}^{n-1}x^{\sigma^i}.$$
The trace maps induces a nondegenerate symmetric $\fq$-bilinear form $\langle \cdot,\cdot\rangle$ on $\fqn$, given by
$$ \langle \alpha,\beta\rangle=\mathrm{Tr}_{q^n/q}(\alpha\beta).$$
The \emph{adjoint} of a $\sigma$-polynomial $f(x)=f_0x+f_1x^{\sigma}+\ldots+f_{n-1}x^{\sigma^{n-1}}$ with respect to the trace bilinear form is 
$$f^\top(x)=\sum_{i=0}^{n-1}{\sigma^{n-i}}(f_i)x^{\sigma^{n-i}},$$
that is the $\sigma$-polynomial satisfying
$$\mathrm{Tr}_{q^n/q}(f(\alpha)\beta)=\mathrm{Tr}_{q^n/q}(\alpha{f}^\top(\beta)), \qquad \mbox{ for every } \alpha,\beta \in \fqn.$$
In this framework, the \emph{adjoint code} $\C^\top$ of a rank-metric code $\C \subseteq \cL_{n,q}$ is
\[ \C^\top= \{{f}^\top(x) \in \cL_{n,q} \colon f(x) \in \C\}. \]

\medskip

\noindent Two rank-metric codes $\C_1, \C_2\subseteq\cL_{n,q}$ are said to be \emph{equivalent} if there exist two invertible $\sigma$-polynomials $f_1(x),f_2(x) \in \cL_{n,q}$ and a field automorphism $\rho \in \mathrm{Aut}(\fqn)$ such that
\[ \C_1=f_1 \circ \C_2^\rho \circ f_2= \{ f_1 \circ g^\rho \circ f_2 \colon g \in \C_2\}, \]
where $g^\rho(x):=\sum_{i=0}^{n-1}\rho(a_i)x^{\sigma^i}$ if $g(x)=\sum_{i=0}^{n-1}a_ix^{\sigma^i}$.

\noindent In addition, a useful tool for studying equivalence of codes is represented by the idealizers. They have been introduced in \cite{liebhold2016automorphism} and used to study equivalence and automorphisms of Gabidulin codes.
Formally, the \emph{left} and \emph{right idealizers} of a rank-metric code
$\C\subseteq\cL_{n,q}$ are defined as
\[L(\C)\coloneqq\{\varphi(x) \in \cL_{n,q} \colon \varphi \circ f \in \C\, \text{for all} \, f \in \C\},\]
\[R(\C)\coloneqq\{\varphi(x) \in \cL_{n,q} \colon f \circ \varphi \in \C\, \text{for all} \, f \in \C\}.\]
Such structures have been also investigated in \cite{lunardon2018nuclei} under the name of middle and right nuclei.
\medskip

\noindent We conclude this section recalling a useful characterization result established in \cite{mcguire2019characterization} (see also \cite{csajbok2019characterization}) for determining the rank of a $\sigma$-polynomial. There, it was shown that it is sufficient to caluclate the rank of a possibly smaller matrix build up with the coefficient of the $\sigma$-polynomial.

\begin{theorem}\cite[Theorem 6]{mcguire2019characterization}
Let $f(x)=\sum_{i=0}^ka_ix^{\sigma^i}$ be an element of $\cL_{n,\sigma}$ with $\sigma$-degree $k$.
Then 
\[ \mathrm{rk}(f)=n-k+\mathrm{rk}(C_f C_f^{\sigma} \cdot \ldots \cdot C_f^{\sigma^{n-1}}-I_k), \]
where 
\[C_f=\left( \begin{matrix} 
0 & 0 & \cdots & 0 & -a_0/a_k \\
1 & 0 & \cdots & 0 & -a_1/a_k \\
0 & 1 & \cdots & 0 & -a_2/a_k \\
\vdots & \vdots & \ddots & \vdots & \vdots \\
0 & 0 & \cdots & 1 & -a_{k-1}/a_k 
\end{matrix}
\right),\]
$C_f^{\sigma^i}$ is the matrix obtained from $C_f$ by applying $\sigma^i$ to each of its entries and $I_k$ is the identity matrix of order $k$.
\end{theorem}

\subsection{Gabidulin and twisted Gabidulin codes}\label{sec:gab_and_tGab}

In this subsection, we survey on the known constructions of $\fqn$-linear MRD codes, represented as subspaces of $\cL_{n,q}$. 

First, we need to specify what is meant by $\fqn$-linearity. The isomorphism described in \eqref{eq:isom_sigma} gives a natural intepretation of this notion as $\fqn$-subspaces of $\cL_{n,q}$. However, this definition is not taking into account equivalence of codes, that is, one may have a rank-metric code $\cC\subseteq \cL_{n,q}$ which is not an $\fqn$-subspace, but it is equivalent to  an $\fqn$-subspace of $\cL_{n,q}$. With this in mind, one can extend the notion of $\fqn$-linearity to codes which \emph{are equivalent to} an $\fqn$-subspace of $\cL_{n,q}$. Sheekey characterized these codes in terms of their idealizers; see \cite[Definition 12]{sheekeysurvey}. Formally, we will say that a rank-metric code $\cC\subseteq \cL_{n,q}$ is \emph{$\fqn$-linear} if $L(\C)$ contains a subring isomorphic to 
\[\mathcal{F}_n:=\{\alpha x \colon \alpha \in \F_{q^n}\}\simeq \F_{q^n}.\]

In \cite{de78}, Delsarte gave the first construction for $\fqn$-linear MRD codes, and few years later, Gabidulin in \cite{ga85a} presented the same class of MRD codes by using linearized polynomials. These codes were then generalized to $\sigma$-polynomials for any generator $\sigma$ of $\Gal(\fqn/\fq)$ by Kshevetskiy and Gabidulin in \cite{Gabidulins}, and they are now known as \emph{Gabidulin codes}.  
Formally, for a given generator $\sigma$ of $\Gal(\fqn/\fq)$ and a positive integer $k\leq n$, the $k$-dimensional \emph{$\sigma$-Gabidulin code} is 
\[\cG_{k,\sigma}=\langle x,x^{\sigma},\ldots,x^{\sigma^{k-1}} \rangle_{\F_{q^n}}.\]
It is easy to see that  $\cG_{k,\sigma}$ is an $\fqn$-linear MRD code  and $L(\cG_{k,\sigma})=R(\cG_{k,\sigma})\simeq \F_{q^n}$; see \cite{liebhold2016automorphism,lunardon2018generalized}.

Five years ago, Sheekey  generalized the family of $\sigma$-Gabidulin codes to what are now known as \emph{twisted Gabidulin codes}. Formally,  the $k$-dimensional \emph{$\sigma$-twisted Gabidulin code} $\mathcal{H}_{k,\sigma}(\eta,h)$ is 
\[\mathcal{H}_{k,\sigma}(\eta,h)=\{a_0x+a_1x^{\sigma}+\ldots+a_{k-1}x^{\sigma^{k-1}}+\sigma^h(a_0)\eta x^{\sigma^{k}} \colon a_i \in \F_{q^n}\}, \]
where $h \in \{0,\ldots,n-1\}$ and $\eta \in \F_{q^n}$ is such that $\N_{q^n/q}(\eta)\neq (-1)^{nk}$. In the same paper Sheekey showed that $\mathcal{H}_{k,\sigma}(\eta,h)$ is an $\F_q$-linear MRD code.
Lunardon, Trombetti and Zhou in \cite{lunardon2018generalized} determined the automorphism group of $\sigma$-twisted Gabidulin codes and studied their equivalence. Moreover,
they also determined their left and right idealizers: if $\eta \neq 0$, then
\[
L(\mathcal{H}_{k,\sigma}(\eta,h))\simeq\F_{q^{\gcd(n,h)}} \quad \mbox{ and } \quad R(\mathcal{H}_{k,\sigma}(\eta,h))\simeq\F_{q^{\gcd(n,k-h)}}.
\]
As a consequence, $\mathcal{H}_{k,\sigma}(\eta,h)$ is $\fqn$-linear if and only if $h=0$.

Further examples of $\fqn$-linear MRD codes can be found in \cite{BZZ,csajbok2018anewfamily,csajbok2018maximum,csajbok2018linearset,MMZ,ZZ} which exist only for $n \in \{6,7,8\}$ and in \cite{longobardi2021large,longobardizanellascatt} which exist for every $n$ even.

\subsection{Equivalence of rank-metric codes}

In this section we recall some known results on the equivalence of rank-metric codes, which will be crucial for showing that the family of codes we are going to introduce is really new. For this purpose, we will use a technique based on applying suitable automorphisms of the Galois group $\Gal(\fqn/\fq)$ to the code and checking the space that their images span. This technique was initiated by Overbeck in \cite{overbeck2008structural} as a structural attack on code-based cryptosystems based on Gabidulin codes. The same idea was elaborated in  \cite{horlemann2017new} for a different purpose: there, Gabidulin codes were partially characterized in terms of the span of the code with its image under a single automorphism, under the additional assumption to have an MRD code. This characterization was later completed in \cite{neri2020systematic}, where the assumption of the code being MRD was dropped. The same strategy was used for the first time to derive code inequivalence results in \cite{puchinger2018construction}. This was later generalized to two automorphisms in \cite{giuzzi2019identifiers}, in order to characterize twisted Gabidulin codes and study their invariants; see also \cite[Section 5.5]{zullo2018linear}. Finally, the general technique for code inequivalence was developed in \cite{neri2019invariants,neri2020equivalence} and was used to give a full characterization of punctured Gabidulin codes and provide a lower bound on their equivalence classes; see also \cite[Chapter 7]{neri2019PhD}.

We first recall a crucial result from \cite{neri2020equivalence}. In the following, for any rank-metric code $\cC\subseteq \cL_{n,q}$ and any automorphism $\sigma \in \Gal(\fqn/\fq)$, we denote by $\sigma(\cC)$ the rank-metric code
$$ \sigma(\cC):=\left\{x^\sigma \circ f(x) : f(x)\in \cC \right\}\subseteq \cL_{n,q}.$$

\begin{proposition}[\textnormal{\cite[Lemma 3.1]{neri2020equivalence}, \cite[Lemma 2]{neri2019invariants}}]\label{prop:equivalence_sum}
 Let $\C_1,\C_2\subseteq \cL_{n,q}$ be two $\fqn$-linear rank-metric codes and let $\sigma_1,\ldots,\sigma_r \in \Gal(\fqn/\fq)$. If $\C_1$ and $\C_2$ are equivalent, then 
 $\sigma_1(\C_1)+\ldots+\sigma_r(\C_1)$ and $\sigma_1(\C_2)+\ldots+\sigma_r(\C_2)$ are also equivalent. 
 In particular,
 $$\dim_{\fqn}(\sigma_1(\C_1)+\ldots+\sigma_r(\C_1))=\dim_{\fqn}(\sigma_1(\C_2)+\ldots+\sigma_r(\C_2)).$$
\end{proposition}
Always in \cite{neri2020equivalence}, the following particular case has been considered. Let $\C$ be an $\fqn$-linear rank-metric code. For a fixed $\sigma$ generator of $\Gal(\fqn/\fq)$, and for a non-negative integer $i$, let
$$s_i^\sigma(\C):=\dim_{\F_{q^n}}(\C+\sigma(\C)+\ldots+\sigma^i(\C)).$$
Then, the sequence $\{s_i^\sigma(\C)\}_{i \in \mathbb N}$ is invariant under code equivalence. This provides a useful criterion to determine when two codes are not equivalent. Of particular interest, the sequences $\{s_i^\sigma(\C)\}_{i\in \mathbb N}$ when $\C$ is a Gabidulin or a twisted Gabidulin code were completely determined and are given by
\begin{equation}\label{eq:gabidulin_sequence} s_i^\sigma(\cG_{k,\sigma})=k+i, \quad \mbox{ for every } i \leq n-k,\end{equation}
\begin{equation}\label{eq:twisted_sequence} s_i^\sigma(\cH_{k,\sigma}(\eta,0))=\begin{cases}k & \mbox { if } i=0 \\
k+i+1 & \mbox{ if } 1\leq i \leq n-k-1\end{cases}
\end{equation}

\medskip

We will also make use of the following notation. 
Let $f(x)=\sum_{i=0}^{n-1}a_ix^{\sigma^i} \in \cL_{n,q}$, the $\sigma$-\emph{support} of $f$ is defined as follows $\mathrm{supp}_{\sigma}(f)=\{i \colon a_i\ne 0 \}$. 
Clearly, if $f(x),g(x) \in \cL_{n,q}$ such that $\mathrm{supp}_{\sigma}(f)\ne \mathrm{supp}_{\sigma}(g)$, then the $f(x)$ and $g(x)$ are $\fqn$-linearly independent.
This definition can be extended to set of linearized polynomials as done in \cite{lunardon2018generalized}: let $\C$ be a rank-metric code in $\cL_{n,q}$,
then the \emph{universal $\sigma$-support} $\cS_{\sigma}(\mathcal{C})$ of $\mathcal{C}$ is defined as the  subset of $\ZZ{n}=\{0,\ldots,n-1\}$
\[ \cS_{\sigma}(\cC)=\bigcup
_{f \in \cC} \supp_{\sigma}(f).
\]
One may immediately notice that if $\mathrm{supp}_{\sigma}(f)$ is not contained in  $\cS_\sigma(\mathcal{C})$, then $f(x)\notin \C$.

\section{A new family of MRD codes}\label{sec:newfamily}

Let us consider the following situation. Let $\sigma$ be a generator of $\mathrm{Gal}(\fqn/\fq)$ and let $s$ be an integer coprime to $n$ and take the extension field of degree $ns$ over $\fq$. Let $\overline{\sigma}$ be an element of $\Gal(\F_{q^{ns}}/\fq)$ that is an  extension of $\sigma$, that is, $\overline{\sigma}|_{\fqn}=\sigma$. We remark that an extension of $\sigma$ always exists over finite fields. Now, observe that $\langle\overline{\sigma}^s\rangle=\Gal(\F_{q^{ns}}/\F_{q^s})$, as shown in the following remark.
\begin{remark}
Suppose that $\sigma:x \in \fqn \mapsto x^{q^i} \in \fqn$, with $\gcd(i,n)=1$. Let $\overline{\sigma}: x \in \F_{q^{ns}} \mapsto x^{q^j} \in \F_{q^{ns}} \in \Gal(\F_{q^{ns}}/\F_{q})$ be an extension of $\sigma$. Since $\overline{\sigma}|_{\fqn}=\sigma$, then $j=i+\ell n$, for a non negative integer $\ell$. Now, $\overline{\sigma}^s: x \in \F_{q^{ns}} \mapsto x^{q^{s(i+\ell n)}} \in \F_{q^{ns}}$, and since $\gcd(i,n)=1$, we get that $\gcd(j,n)=\gcd(i+\ell n,n)=1$ and hence $\langle\overline{\sigma}^s\rangle=\Gal(\F_{q^{ns}}/\F_{q^s})$.
\end{remark}

Define the maps $\Psi_{\sigma,s}$ and $\Phi_{\sigma,s}$ as
$$ \begin{array}{rccl}
    \Psi_{\sigma,s} : &\cL_{n,\sigma} & \longrightarrow & \cL_{n,\overline{\sigma}^s}  \\
     &\sum_i f_ix^{\sigma^i}  &  \longmapsto & \sum_i f_ix^{\overline{\sigma}^{si}},
\end{array}$$
$$ \begin{array}{rccl}   
   \Phi_{\sigma,s} : &\cL_{n,\sigma} & \longrightarrow & \cL_{n, \sigma^s}  \\
     &\sum_i f_ix^{\sigma^i}  &  \longmapsto & \sum_i f_ix^{\sigma^{si}}.
\end{array} $$
Observe that the map $\Phi_{\sigma,s}$ is a bijection, while the map $\Psi_{\sigma,s}$ is only injective. 

\begin{theorem}\label{th:rankfs}
Let $s$ be an integer coprime to $n$ and take the extension field of degree $ns$ over $\fq$ and let $\sigma$ be a generator of $\mathrm{Gal}(\fqn/\fq)$.
Let $f(x) \in \cL_{n,\sigma}$. Then $\rk(\Psi_{\sigma,s}(f))=\rk(\Phi_{\sigma,s}(f))$.
\end{theorem}
\begin{proof}
Denote by $k$ the $\sigma$-degree of $f(x)$. Then
\[\mathrm{rk}(\Phi_{\sigma,s}(f))=n-k+\mathrm{rk}_{\fqn}(C_{\Phi_{\sigma,s}(f)} C_{\Phi_{\sigma,s}(f)}^{\sigma^s} \cdot \ldots \cdot C_{\Phi_{\sigma,s}(f)}^{\sigma^{s(n-1)}}-I_k)\]
and
\[ \mathrm{rk}(\Psi_{\sigma,s}(f))=n-k+\mathrm{rk}_{\F_{q^{ns}}}(C_{\Psi_{\sigma,s}(f)} C_{\Psi_{\sigma,s}(f)}^{\overline{\sigma}^s} \cdot \ldots \cdot C_{\Psi_{\sigma,s}(f)}^{\overline{\sigma}^{s(n-1)}}-I_k). \]
Since $C_{\Phi_{\sigma,s}(f)}=C_{\Psi_{\sigma,s}(f)} \in \fqn^{k\times k}$ and $\overline{\sigma}\in \mathrm{Gal}(\F_{q^{ns}}/\F_{q})$ such that $\overline{\sigma}|_{\fqn}=\sigma$, it follows that 
\[ C_{\Phi_{\sigma,s}(f)} C_{\Phi_{\sigma,s}(f)}^{\sigma^s} \cdot \ldots \cdot C_{\Phi_{\sigma,s}(f)}^{\sigma^{s(n-1)}}=C_{\Psi_{\sigma,s}(f)} C_{\Psi_{\sigma,s}(f)}^{\overline{\sigma}^s} \cdot \ldots \cdot C_{\Psi_{\sigma,s}(f)}^{\overline{\sigma}^{s(n-1)}}, \]
and hence $\rk(\Psi_{\sigma,s}(f))=\rk(\Phi_{\sigma,s}(f))$.
\end{proof}

As a natural consequence one can get the following result on rank-metric codes.

\begin{corollary}\label{cor:s}
Let $s$ be an integer coprime to $n$.
Let $\mathcal{C}$ be an $\F_p$-linear rank-metric code of $\cL_{n,\sigma}$ and let $\overline{\sigma}\in \mathrm{Gal}(\F_{q^{ns}}/\F_{q})$ be an extension of $\sigma$. 
Suppose that $\Psi_{\sigma,s}(\C)\subseteq \cL_{n,\overline{\sigma}^s}$ has minimum distance $d$. Then $\Phi_{\sigma,s}(\mathcal{C})\subseteq \cL_{n,\sigma^s}$ has minimum distance $d$.
\end{corollary} 

In the spirit of \cite[Theorem 3.2]{lunardon2018generalized}, Corollary \ref{cor:s} can be specialized to MRD codes as follows.

\begin{corollary}\label{cor:s2}
Let $s$ be an integer coprime to $n$.
Let $\mathcal{C}$ be an $\F_p$-linear MRD code of $\cL_{n,\sigma}$ with minimum distance $d$ and let $\overline{\sigma}\in \mathrm{Gal}(\F_{q^{ns}}/\F_{q})$ be an extension of $\sigma$.
Assume that $\Psi_{\sigma,s}(\C) \subseteq \cL_{n,\overline{\sigma}^s}$ has minimum distance at least $d$. Then $\Phi_{\sigma,s}(\mathcal{C})\subseteq \cL_{n,\sigma^s}$ is an MRD code.
\end{corollary}

\begin{remark}
 The assumption on $\Psi_{\sigma,s}(\C) \subseteq \cL_{n,\overline{\sigma}^s}$ in Corollary \ref{cor:s2} is satisfied by all the known examples of MRD codes, that is the minimum distance of $\C\subseteq \cL_{n,\sigma}$ coincides with the minimum distance of $\Psi_{\sigma,s}(\C)$, for any $s$. So, one could use Corollary \ref{cor:s2} and the fact that $\theta$-Gabidulin codes and $\theta$-twisted Gabidulin codes are MRD when $\theta$ is the $q$-Frobenius automorphism, to prove that $\sigma$-Gabidulin and $\sigma$-twisted Gabidulin codes are MRD, for every generator $\sigma$ of $\Gal(\fqn/\fq)$. This is indeed the technique used in \cite{Gabidulins} for  $\sigma$-Gabidulin codes and in \cite{lunardon2018generalized} for $\sigma$-twisted Gabidulin codes.
\end{remark}

Recently, in \cite{longobardi2021large} and in \cite{longobardizanellascatt} two families of MRD codes that exist for infinitely many values of $n$ were presented. We recall them via the following two theorems.

\begin{theorem}[\textnormal{\cite[Theorem 2.4]{longobardizanellascatt}}]\label{th:scattLZ}
Let $n=2t$, $t \geq 3$, $q$ be an odd prime power and let $\sigma$ be a generator of $\mathrm{Gal}(\fqn/\fq)$.
If $t$ is even, or $t$ is odd and $q\equiv 1 \pmod{4}$, then the rank-metric code
\[\C=\langle x,\psi(x)  \rangle_{\fqn},\]
with $\psi(x)=x^\sigma+x^{\sigma^{t-1}}-x^{\sigma^{t+1}}+x^{\sigma^{2t-1}}$, is an MRD code.
\end{theorem}

\begin{theorem}[\textnormal{\cite[Theorem 3.1]{longobardi2021large}}] \label{th:scatteredLMTZ}
Let $n=2t$, $t \geq 3$, let $q$ be an odd prime power and let $\theta \colon x \in \fqn \longmapsto x^q \in \fqn$. For each $h \in \F_{q^n} \setminus \F_{q^t}$ such that $\N_{q^n/q^t}(h)=-1$,
the rank-metric code
\[ \C_{h,t}=\langle x,\psi_{h,t}(x) \rangle_{\fqn}, \]
where 
\[
    \psi_{h,t}(x)=x^\theta+x^{\theta^{t-1}}+h\theta(h)x^{\theta^{t+1}}+h\theta^{-1}(h^{-1})x^{\theta^{2t-1}} \in \cL_{n,q} ,
\]
is an MRD code.
\end{theorem}

In the following result we extend the above constructions, by using Corollary \ref{cor:s2}.

\begin{theorem}\label{th:newfamily}
Let $n=2t$, $t \geq 3$, $q$ be an odd prime power and let $\sigma$ be a generator of $\mathrm{Gal}(\fqn/\fq)$. For any $h \in \F_{q^n}$ such that $\N_{q^n/q^t}(h)=-1$, the rank-metric code
$$\mathcal{C}_{h,t,\sigma}= \langle x, \psi_{h,t,\sigma}(x) \rangle_{\F_{q^n}},$$
with
\[ \psi_{h,t,\sigma}(x)=x^\sigma +x^{\sigma^{t-1}}+h\sigma(h)x^{\sigma^{t+1}}+h\sigma^{-1}(h^{-1})x^{\sigma^{2t-1}}, \]
is an MRD code.
\end{theorem}

\begin{proof}
First assume that $h \notin \F_{q^t}$.
Let $\sigma\colon x \in \fqn \mapsto x^{q^s}\in \fqn$, with $\gcd(s,n)=1$, and consider $\Psi_{\theta,s}(\C_{h,t,\theta})\subseteq \cL_{n,\overline{\sigma}}$, where $\theta\colon x \in \fqn \mapsto x^{q}\in \fqn$ and $\overline{\sigma}\colon x \in \F_{q^{ns}} \mapsto x^{q^s}\in \F_{q^{ns}}$.
We now show that $\Psi_{\theta,s}(\C_{h,t,\theta})$ is an MRD code.
First note that $h \notin \F_{q^{st}}$. Indeed if $h \in \F_{q^{st}}$, then $h \in \F_{q^{2t}} \cap \F_{q^{st}} = \F_{q^t}$, a contradiction. Moreover, since $\overline{\sigma}$ is a generator of $\mathrm{Gal}(\F_{q^{ns}}/\F_{q^s})$ then $\overline{\sigma}^t(h)h=\N_{q^n/q^t}(h)=-1$. 
By applying Theorem \ref{th:scatteredLMTZ}, $\psi_{\theta,s}(\C_{h,t,\theta})$ is a rank-metric code of $\cL_{n,\overline{\sigma}}$ having minimum distance $n-1$. Corollary \ref{cor:s2} implies that $\Phi_{\theta,s}(\C_{h,t,\theta})=\C_{h,t,\sigma}$ is an MRD code contained in $\cL_{n,\sigma}$.
If $h$ is in $\F_{q^t}$, then $\N_{q^n/q^t}(h)=h^2=-1$ and hence $h \in \F_{q^2}$.
So, if $h \in \fq$ then
\[\psi_{h,t,\sigma}(x)=x^\sigma+x^{\sigma^{t-1}}-x^{\sigma^{t+1}}+x^{\sigma^{2t-1}},\]
and if $h \in \F_{q^2}\setminus \fq$ (and hence $t$ is even) then
\[\psi_{h,t,\sigma}(x)=x^\sigma+x^{\sigma^{t-1}}+x^{\sigma^{t+1}}-x^{\sigma^{2t-1}}=x^\theta+x^{\theta^{t-1}}-x^{\theta^{t+1}}+x^{\theta^{2t-1}},\]
where $\theta=\sigma^{t-1}$, so that the rank-metric code $\C_{h,t,\sigma}$ is MRD because of Theorem \ref{th:scattLZ}.
\end{proof}

When choosing $\sigma$ as $x\in \fqn \mapsto x^{q}\in \fqn$, the MRD codes $\C_{h,t,\sigma}$ coincide with the MRD codes in Theorem \ref{th:scatteredLMTZ}.
Whereas, if $h$ is in $\F_{q^t}$ in Theorem \ref{th:newfamily}, then the MRD code $\C_{h,t,\sigma}$ coincide with those in Theorem \ref{th:scattLZ}, as already pointed out in Theorem \ref{th:newfamily}.

\begin{remark}
Let $V$ be a $2$-dimensional vector space over $\fqn$ and let $\Lambda=\PG(V,\F_{q^n})$.
Let $U$ be an $\fq$-subspace of $V$ of dimension $k$, then the set of points
\[ L_U=\{\la {\bf u} \ra_{\mathbb{F}_{q^n}} : {\bf u}\in U\setminus \{{\bf 0} \}\}\subseteq \Lambda \]
is said to be an $\fq$-\emph{linear set of rank $k$}.
The $\fq$-linear set $L_U$ is called \emph{scattered} if $|L_U|=\frac{q^k-1}{q-1}$. 
We refer to \cite{lavrauw2015field} and \cite{polverino2010linear} for comprehensive references on linear sets.
Let $L_U$ be an $\fq$-linear set of rank $n$ in $\Lambda$. Since ${\rm P\Gamma L}(1,q^n)$ is $3$-transitive on $\PG(1,q^n)$, we can suppose that (up to ${\rm P\Gamma L}(1,q^n)$-equivalence) $L_U$ does not contain the point $\langle(0,1)\rangle_{\fqn}$, so that there exists $f(x) \in \cL_{n,q}$
 such that $L_U$ is ${\rm P\Gamma L}(1,q^n)$-equivalent to
\[ L_f=\{ \langle (x,f(x)) \rangle_{\fqn} \colon x \in \fqn^* \}. \]
Sheekey in \cite{sheekey2016new} called a $\sigma$-polynomial $f(x)$ \emph{scattered} if $L_f$ turns out to be a scattered $\fq$-linear set.
In \cite{sheekey2016new}, it was shown a correspondence between scattered linear sets of $\mathrm{PG}(1,q^n)$ and $\fqn$-linear MRD codes in $\cL_{n,q}$, which has been later generalized in \cite{csajbok2017maximum,lunardon2017mrd,polverino2020connections,sheekeyVdV,zini2020scattered}.
In particular, the linear set $L_f$ is scattered if and only if $\mathcal{C}_f=\langle x,f(x)\rangle_{\fqn}$ is an MRD code. 
So, under the assumption of Theorem \ref{th:newfamily}, we have proved that the  $\psi_{h,t,\sigma}(x)$ is a scattered polynomial.
In \cite{longobardi2021large,longobardizanellascatt}, the results are stated using the terminology of linear sets.
\end{remark}

In the next section we study in details the equivalence issue of the codes in Theorem \ref{th:newfamily} with themselves and with the other known $\fqn$-linear MRD codes, proving that the family of $\C_{h,t,\sigma}$'s contains new MRD codes.

\section{Study of the equivalence of the new family}\label{sec:equivalence}

This section is dedicated to the equivalence of the new family of codes $\cC_{h,t,\sigma}$. We  first discuss their novelty,  showing that they are not equivalent to any family of $\fqn$-linear MRD codes known. Later, we  study their equivalence classes, characterizing exactly when two codes of the form $\cC_{h,t,\sigma}$ and $\cC_{k,t,\theta}$ are equivalent. From these results we can  deduce the exact number of equivalence classes of such codes. However, since this formula is a bit involved, we also provide a more effective lower bound.

We fix here the following notation that will be kept for the whole section. From now on, we consider $\sigma$ to be a generator of $\Gal(\fqn/\fq)$ and $h\in\fqn$ to be such that $\sigma^t(h)h=-1$. Furthermore,  we fix $n=2t$, for some positive integer $t \geq 5$. Notice that this assumption on $t$ is taken in order to ease the computations, even though the codes $\cC_{h,t,\sigma}$ are MRD also for $t=3,4$. When
$t=3$ the inequivalence with the other known $\fqn$-linear MRD codes such as Gabidulin codes, twisted Gabidulin codes and those in \cite{csajbok2018anewfamily,csajbok2018maximum,csajbok2018linearset,MMZ,ZZ} has been proved in \cite[Section 4]{BZZ}.
However when
$t \in \{3,4\}$ the computations of this section become more complicated, since some of the arguments that we are going to use do not work. We warn the reader that  for these two cases some of the results presented might not be true.

We first start with an auxiliary result which will be widely used in this section.

\begin{lemma}\label{lem:auxiliary}
 Let $h\in\fqn$ be such that $\N_{q^n/q^t}(h)=\sigma^t(h)h=-1$. Then $\sigma^2(h)h\neq 1$.
\end{lemma}

\begin{proof}
 Since $q$ is odd, then necessarily $h \neq \pm 1$. Assume now by contradiction that $\sigma^2(h)h=1$, or equivalently $\sigma^2(h)=h^{-1}$, and $\sigma^t(h)h=-1$.
 Then
 $$h=\sigma^n(h)=\underbrace{\sigma^2\circ\ldots\circ\sigma^2}_{t \mbox{ \tiny{times} }}(h)=\begin{cases}
 h & \mbox{ if } t \mbox{ is even,}\\
 h^{-1} & \mbox{ if } t \mbox{ is odd,}
 \end{cases}$$
 which implies that $t$ is even. Moreover, we have $\sigma^4(h)=h$ and let $t\equiv t' \pmod{4}$ with $t' \in\{0,2\}$, hence
$$-1=\sigma^t(h)h=( \sigma^{t'}\circ(\sigma^4\circ\ldots\circ\sigma^4))(h) h=\sigma^{t'}(h)h=\begin{cases}
 \sigma^2(h)h & \mbox{ if } t \equiv 2 \, \pmod{4},\\
 h^2 & \mbox{ if } t \equiv 0 \, \pmod{4}.
 \end{cases}$$
 Since $q$ is odd, $1\neq -1$ and we get a contradiction if $t\equiv 2\, \pmod{4}$. If instead $t \equiv 0 \, \pmod{4}$, we have $h^2=-1$, and therefore $h \in \F_{q^2}$, which in turn implies $\sigma^2(h)=h$. Thus, we have simultaneously $h^2=1$ and $h^2=-1$, which yields a contradiction.
\end{proof}

\subsection{Inequivalence with Gabidulin and twisted Gabidulin codes}

Here we compare the construction of the family $\cC_{h,t,\sigma}$ with the two most prominent family of $\fqn$-linear MRD codes, namely Gabidulin  and twisted Gabidulin codes. 

\begin{proposition}\label{prop:not_Gabidulin}
 Let $\cG_{2,\theta}$ be a $2$-dimensional Gabidulin code. Then $\cG_{2,\theta}$ is not equivalent to $\cC_{h,t,\sigma}$, for any  $h$ such that $\sigma^t(h)h=-1$ and any generators $\sigma,\theta$ of $\Gal(\fqn/\fq)$. 
\end{proposition}

\begin{proof}
 Since both $\sigma$ and $\theta$ generate $\Gal(\fqn/\fq)$, then we can write $\theta=\sigma^s$ for some $s$ coprime with $n$. Moreover, $\cG_{2,\theta}$ is equivalent to $\cG_{2,\theta^{-1}}$, hence, without loss of generality we can suppose $0<s<t$. Suppose that $\cG_{2,\sigma^s}$ and $\cC_{h,t,\sigma}$ are equivalent, then by Proposition \ref{prop:equivalence_sum} and \eqref{eq:gabidulin_sequence} we must have
 $s_1^{\sigma^s}(\cC_{h,t,\sigma})=3$, that is $\dim_{\fqn}(\cD)=3$, where
  $\cD:=\cC_{h,t,\sigma}+\sigma^{s}(\cC_{h,t,\sigma})$. We have
  $$\cD=\langle x, x^{\sigma^s}, \psi_{h,t,\sigma}(x), x^{\sigma^s}\circ \psi_{h,t,\sigma}(x) \rangle_{\fqn}.$$
  Notice that the first three generators are clearly linearly indpendent, since the first two have disjoint $\sigma$-supports and the $\sigma$-support of $\psi_{h,t,\sigma}(x)$ is not contained in $\{0,s\}$. Thus, we must have $x^{\sigma^s}\circ \psi_{h,t,\sigma}(x)\in \langle x,x^{\sigma^s},\psi_{h,t,\sigma}(x)\rangle_{\fqn}$, that implies 
  $$\supp_\sigma(x^{\sigma^s}\circ \psi_{h,t,\sigma}(x))=\{s+1,t+s-1,t+s+1,s-1\}\subseteq \{0,1,s,t-1,t+1,2t-1\}.$$
  However, due to the restriction on $s$ being coprime with $2t$, this can never happen.
\end{proof}

\begin{theorem}\label{th:no_twisted}
  Let $\cH_{\eta,\theta}:=\cH_{2,\theta}(\eta,0)$ be a $2$-dimensional $\fqn$-linear twisted Gabidulin code, and let $t\geq 5$. Then $\cH_{\eta,\theta}$ is not equivalent to $\cC_{h,t,\sigma}$, for any admissible choice of $h$, $\eta$ and any generators $\sigma, \theta$ of $\Gal(\fqn/\fq)$. 
\end{theorem}

\begin{proof}
 Since both $\sigma$ and $\theta$ generate $\Gal(\fqn/\fq)$, then we can write $\theta=\sigma^s$ for some $s$ coprime with $n$. Moreover, $\cH_{\eta,\theta}$ is equivalent to $\cH_{\eta',\theta^{-1}}$, hence, without loss of generality we can suppose $0<s<t$. So, from now on, we are assuming
 \begin{equation}\label{eq:conditions_on_s}
     1\leq s \leq t-1, \qquad \gcd(s,2t)=1
 \end{equation}
 Suppose that $\cH_{\eta,\sigma^s}$ and $\cC_{h,t,\sigma}$ are equivalent, then by Proposition \ref{prop:equivalence_sum} and \eqref{eq:twisted_sequence} we must have 
 $s_1^{\sigma^s}(\cC_{h,t,\sigma})=4$ and $s_2^{\sigma^s}(\cC_{h,t,\sigma})=5$. The first assertion has to be true. In fact, if  $s_1^{\sigma^s}(\cC_{h,t,\sigma})=2$, then $\cC_{h,t,\sigma}$ has rank one elements and cannot be MRD. Moreover, if  $s_1^{\sigma^s}(\cC_{h,t,\sigma})=3$, since $\cC_{h,t,\sigma}$ is MRD, then it must be a Gabidulin code (see e.g. \cite[Theorem 4.8]{horlemann2017new}). However,  by Proposition \ref{prop:not_Gabidulin}, this is not possible. Hence, we deduce that $s_1^{\sigma^s}(\cC_{h,t,\sigma})=4$. Assume now that $s_2^{\sigma^s}(\cC_{h,t,\sigma})=5$, and define
 $$\cD:=\cC_{h,t,\sigma}+\sigma^{s}(\cC_{h,t,\sigma}).$$
 
 We first show that $x^{\sigma^{2s}} \notin \cD$. Assume on the contrary that $x^{\sigma^{2s}} \in \cD$. This implies that $2s\in\cS_{\sigma}(\cD)=\{0,1,s-1,s,s+1,t-1,t+1,t+s-1,t+s+1,-1\}$. Recalling our assumptions \eqref{eq:conditions_on_s} on $s$, it is easy to observe that we can only have
 $$2s\in \{s+1,t-1,t+1,t+s-1\}.$$
   
  \medskip
  
  \noindent\underline{\textbf{Case I:} $2s=s+1$}. If this happens then $s=1$. This translates in 
 $$x^{\sigma^2}\in \langle x, x^{\sigma}, x^{\sigma^{t-1}}+h\sigma(h)x^{\sigma^{t+1}}+h\sigma^{-1}(h^{-1})x^{\sigma^{2t-1}}, x^{\sigma^2}+x^{\sigma^{t}}+\sigma(h)\sigma^2(h)x^{\sigma^{t+2}}\rangle_{\fqn}=\cD.$$
 It is easy to see that this is not possible. Indeed, the only generator of $\cD$ whose $\sigma$-support contains $2$ is $f(x)=x^{\sigma^2}+x^{\sigma^{t}}+\sigma(h)\sigma^2(h)x^{\sigma^{t+2}}$, and its $\sigma$-support is disjoint from the $\sigma$-supports of the other generators.
   
  \medskip
  
  \noindent\underline{\textbf{Case II:} $2s=t-1$}. Then
  $$x^{\sigma^{2s}}\in \langle x, x^{\sigma^{s}}, x^\sigma+x^{\sigma^{2s}}+h\sigma^{-1}(h^{-1})x^{\sigma^{2s+1}}, x^{\sigma^{s+1}}+x^{\sigma^{s-2}}+\sigma^s(h)\sigma^{s-1}(h^{-1})x^{\sigma^{s-1}}\rangle_{\fqn}=\cD.$$
 The only generator of $\cD$ whose $\sigma$-support contains $2s$ is $f(x)=x^\sigma+x^{\sigma^{2s}}+h\sigma^{-1}(h^{-1})x^{\sigma^{2s+1}}$, and its $\sigma$-support is disjoint from the $\sigma$-supports of the other generators, due to the conditions \eqref{eq:conditions_on_s} that imply $s\geq 3$. 
   
  \medskip
  
  \noindent \underline{\textbf{Case III and IV:}}  Analogously one can exclude these two cases. 
  \medskip
  
  \noindent At this point, since we are assuming $s_2^{\sigma^s}(\cC_{h,t,\sigma})=5$, we must have
  $$g(x):=x^{\sigma^{2s+1}}+x^{\sigma^{2s+t-1}}+\sigma^{2s}(h)\sigma^{2s+1}(h)x^{\sigma^{2s+t+1}}+\sigma^{2s}(h)\sigma^{2s-1}(h^{-1})x^{\sigma^{2s-1}} \in \cD',$$
  where $\cD':=\cD+\langle x^{\sigma^{2s}}\rangle_{\fqn}$.
  This is equivalent to say $\supp_{\sigma}(g)\subseteq \cS_\sigma(\cD')$, i.e.
  \begin{equation}\label{eq:containment}\{2s-1,2s+1,2s+t-1,2s+t+1\} \subseteq \{0,1,-1,t-1,t+1,s,s-1,s+1,s+t-1,s+t+1,2s\}
  \end{equation}
  and in particular $2s+1\in \cS_\sigma(\cD)$.
  One can observe that, due to the conditions \eqref{eq:conditions_on_s}, we can only have
  $$ 2s+1 \in \{-1,t-1,t+s-1\}.$$
  \medskip
  
  \noindent \underline{\textbf{Case I:} $2s+1=-1$}. This means $s=t-1$ and \eqref{eq:containment} becomes
  $$\{-3,-1,t-3,t-1\}\subseteq \{0,1,-1,t-2,t-1,t,t+1,-2\}.$$
  However, since $t\geq 5$, this is not possible. 
  
  \medskip
  
  \noindent \underline{\textbf{Case II:} $2s+1=t-1$}. In this case \eqref{eq:containment} can be written as
  $$\{-3,-1,2s-1,2s+1\}\subseteq\{0,1,-1,s,s-1,s+1,2s,2s+1,2s+3,3s+1,3s+3 \} $$
  Since $t\geq 5$ and $2s+1=t-1$, and because of \eqref{eq:conditions_on_s}, we have $s\geq 3$ and $t\geq 8$. With this range of parameters it is immediate to see that $-3 \notin \cS_\sigma(\cD')$, so also this case is not possible. 
  
  \medskip
  
  \noindent\underline{\textbf{Case III:} $2s+1=t+s-1$}. This means $s=t-1$ and \eqref{eq:containment} becomes
  $$\{-5,-3,t-5,t-3\}\subseteq \{0,1,-1,t-3,t-2,t-1,t+1,-4,-3,\}.$$
  Since we are assuming $t\geq 5$, it is easy to see that $-5$ and $t-5$ do no belong to $\cS_\sigma(\cD')$. 
  \medskip
  
 \noindent This finally shows that $s_2^{\sigma^s}(\cC_{h,t,\sigma})=6$, and thus $\cC_{h,t,\sigma}$ and $\cH_{\eta,\sigma^s}$ cannot be equivalent due to Proposition \ref{prop:equivalence_sum}.
\end{proof}
 \begin{remark}
  Note that neither the equivalence of a code $\cC_{h,t,\sigma}$ with any other twisted Gabidulin code $\cH_{2,\theta}(\eta,h)$ is possible. Indeed, as discussed in Section \ref{sec:gab_and_tGab}, if $h \neq 0$   then the left idealizer of $\cH_{2,\theta}(\eta,h)$ is isomorphic to a finite field of cardinality $q^{\gcd(h,n)}$. Therefore, $\cH_{2,\theta}(\eta,h)$ is not an $\fqn$-linear code, while $\cC_{h,t,\sigma}$ is $\fqn$-linear and has left idealizer isomorphic to $\fqn$; by  \cite[Proposition 4.1]{lunardon2018nuclei} these codes cannot be equivalent. 
 \end{remark}

\subsection{Equivalence among themselves}

In this section we investigate the equivalence issue for codes within the family introduced in Theorem \ref{th:newfamily}.
We start with the following remark.
\begin{remark}\label{rem:h}
Let $\rho \in \mathrm{Aut}(\fqn)$. One can easily check that $\mathrm{N}_{q^n/q^t}(\rho(h))=\rho(\mathrm{N}_{q^n/q^t}(h))=\rho(-1)=-1$ and $(\C_{h,t,\sigma})^\rho=\C_{\rho(h),t,\sigma}$. Hence we can immediately derive that
 $\C_{h,t,\sigma}$ and $\C_{\rho(h),t,\sigma}$ are  equivalent.  
\end{remark}

\begin{theorem}\label{th:equivspecialcase}
Let $t\geq 5$ and consider $\mathcal{C}_{h,t,\sigma}$ and $\C_{k,t,\sigma^s}$. Then the following hold:
\begin{itemize}
    \item[\textbf{I.}] If $s = 1$, then $\mathcal{C}_{h,t,\sigma}$ and $\C_{k,t,\sigma^s}$ are equivalent if and only if there exists $\rho \in \mathrm{Aut}(\fqn)$ such that
        \[ \rho(h)=\left\{ \begin{array}{ll} \pm k, & \text{if}\,\, t \not\equiv 2 \pmod{4},\\ l k,\,\,\text{where}\,l^{q^2+1}=1, & \text{if}\,\, t \equiv 2 \pmod{4}. \end{array} \right. \]
    \item[\textbf{II.}] If $s = -1$, $\mathcal{C}_{h,t,\sigma}$ and $\C_{k,t,\sigma^s}$ are equivalent if and only if there exists $\rho \in \mathrm{Aut}(\fqn)$
    \[ \rho(h)=\left\{ \begin{array}{ll} \pm k^{-1}, & \text{if}\,\, t \not\equiv 2 \pmod{4},\\ l k^{-1},\,\text{where}\,l^{q^2+1}=1, & \text{if}\,\, t \equiv 2 \pmod{4}. \end{array} \right. \]
    \item[\textbf{III.}] If $s = t-1$, with $t$ even, then $\mathcal{C}_{h,t,\sigma}$ and $\C_{k,t,\sigma^s}$ are equivalent if and only if there exists $\rho \in \mathrm{Aut}(\fqn)$ such that
    \[ \rho(h)=\left\{ \begin{array}{ll} \pm k, & \text{if}\,\, t \not\equiv 2 \pmod{4},\\ l k,\,\text{where}\,l^{q^2+1}=1, & \text{if}\,\, t \equiv 2 \pmod{4}. \end{array} \right. \]
    \item[\textbf{IV.}] If $s = t+1$, with $t$ even, then $\mathcal{C}_{h,t,\sigma}$ and $\C_{k,t,\sigma^s}$ are equivalent if and only if there exists $\rho \in \mathrm{Aut}(\fqn)$ such that
    \[ \rho(h)=\left\{ \begin{array}{ll} \pm k^{-1}, & \text{if}\,\, t \not\equiv 2 \pmod{4},\\ l k^{-1},\,\text{where}\,l^{q^2+1}=1, & \text{if}\,\, t \equiv 2 \pmod{4}. \end{array} \right. \]
\end{itemize}
\end{theorem}
\begin{proof}
\underline{\textbf{Case I:} $s=1$}.
Suppose that  $\mathcal{C}_{h,t,\sigma}$ and $\C_{k,t,\sigma}$ are equivalent, so there exists an isometry $\varphi$ such that $\varphi(\C_{h,t,\sigma})=\C_{k,t,\sigma}$.
So, $\varphi(p(x))=f_1 \circ p(x)^\rho \circ f_2(x)$, for $f_1(x)$ and $f_2(x)$ two invertible $\sigma$-polynomials in $\cL_{n,q}$ and $\rho \in \mathrm{Aut}(\fqn)$. 
By \cite[Proposition 3.8]{sheekeyVdV}, it follows that we may assume $f_1(x)=x^{\sigma^i}$ for some $i \in \{0,\ldots,n-1\}$. 
We may also assume $\rho$ to be the identity map and $i=0$, since by Remark \ref{rem:h} the code $\C_{h,t,\sigma}$ is equivalent to $\C_{\sigma^i(\rho(h)),t,\sigma}$.
Since $\varphi(x) \in \C_{k,t,\sigma}$ then $\varphi(x)=ax+b\psi_{k,t,\sigma}(x)$, which implies that $f_2(x)=ax+b\psi_{k,t,\sigma}(x)$.
Also, since $\varphi( \psi_{h,t,\sigma}(x)) \in \C_{k,t,\sigma}$ then $\varphi(\psi_{h,t,\sigma}(x))=cx+d\psi_{k,t,\sigma}(x)$.
In particular, one gets
\begin{equation}\label{eq:s=1prima} cx+d\psi_{k,t,\sigma}(x)=\psi_{h,t,\sigma}(ax)+\psi_{k,t,\sigma}(b\psi_{h,t,\sigma}(x)). \end{equation}

\noindent Arguing as in \cite[Case 1 and Case 2, Proof of Theorem 4.2]{longobardi2021large} one gets that $h=lk$, with $l=\pm 1$ if $t\not\equiv 2\pmod{4}$ and $l\in \F_{q^4}$ such that $l^{q^2+1}=1$ if $t\equiv 2\pmod{4}$.
Conversely, assume that $h=lk$, with $l=\pm 1$ if $t\not\equiv 2\pmod{4}$ and $l\in \F_{q^4}$ such that $l^{q^2+1}=1$ if $t\equiv 2\pmod{4}$.
If $t \not \equiv 2\pmod{4}$ then the two codes $\C_{h,t,\sigma}$ and $\C_{k,t,\sigma}$ coincide, and hence they are equivalent.
Now, suppose that $t \equiv 2\pmod{4}$.
Choose $a \in \F_{q^4}$ such that 
\[ \sigma^3(a)\sigma(a^{-1})=l\sigma(l), \]
and $d=\sigma(a)$.
As in \cite[Case 1, Proof of Theorem 4.2]{longobardi2021large}, one can show that
\begin{equation}\label{eq:systems=1}
\begin{cases}
d=\sigma(a)=\sigma^{t-1}(a),\\
dh\sigma(h)=\sigma^{t+1}(a)k\sigma(k),\\
dh\sigma^{-1}(h^{-1})=\sigma^{-1}(a)k\sigma^{-1}(k^{-1}),
\end{cases}
\end{equation}
so that
\[ d\psi_{h,t,\sigma}(x)=\psi_{k,t,\sigma}(ax). \]
Consider $\tau(x)=ax$, then
\[ \C_{k,t,\sigma}\circ \tau=\langle ax,d\psi_{h,t,\sigma}(x)\rangle_{\fqn}=\C_{h,t,\sigma}, \]
and the two codes $\C_{k,t,\sigma}$ and $\C_{h,t,\sigma}$ are equivalent. 

\medskip

\noindent \underline{\textbf{Case II:} $s=-1$}.
Our aim is first to prove that $\C_{k,t,\sigma^{-1}}$ is equivalent to $\C_{k^{-1},t,\sigma}$, and then we use the case $s=1$.
So, now we need find $b,c\in \fqn^*$ such that 
\begin{equation}\label{eq:s=-1map}
cx=\psi_{k,t,\sigma^{-1}}(b\psi_{k^{-1},t,\sigma}(x)),
\end{equation}
which is enough to prove that $\C_{k,t,\sigma^{-1}}$ is equivalent to $\C_{k^{-1},t,\sigma}$.
Indeed, suppose we have found such $b$ and $c$, then the map $\psi_{k,t,\sigma^{-1}}(x)$ is invertible, otherwise we would get a contradiction to \eqref{eq:s=-1map}, and $\psi_{k,t,\sigma^{-1}}(\C_{k^{-1},t,\sigma})=\langle \psi_{k,t,\sigma^{-1}}(x),cx\rangle_{\fqn}= \C_{k,t,\sigma^{-1}}$.
Equation \eqref{eq:s=-1map} can be written as follows:
\begin{align*}
    cx= &\, (\sigma(k)k\sigma^{t+1}(b)+\sigma(k^{-1})\sigma^2(k)\sigma(b))x^{\sigma^{2}} \\
 & +(\sigma^{t-1}(b)-\sigma^{-1}(k^{-2})k\sigma^{t-2}(k)\sigma^{-1}(b))x^{\sigma^{t-2}}  \\
& +(\sigma(k)k \sigma^{t+1}(b)+ \sigma^{-1}(k^{-1})k\sigma^{-1}(b)+\sigma(k^{-1})k^{-1}\sigma(b)-\sigma^{t-1}(k^{-1})k^{-1}\sigma^{t-1}(b))x^{\sigma^{t}}  \\
& +(\sigma^{-1}(k^{-1})k\sigma^{-1}(b)-\sigma^{t-1}(k^{-1})\sigma^{-2}(k)\sigma^{t-1}(b))x^{\sigma^{2t-2}} \\
& +(\sigma(b)+\sigma^{t-1}(b)+\sigma(k^2)k^2\sigma^{t+1}(b)+\sigma^{-1}(k^{-2})k^2 \sigma^{-1}(b)) x \\
& +(\sigma(b)+\sigma^2(k)k\sigma^2(k^{-1})\sigma^{t+1}(b))x^{\sigma^{t+2}}.
\end{align*}

Now, we choose
\[ b=\frac{1}{k^{-1}\sigma^{-1}(k^{-1})-k^{-1}\sigma(k)}. \]
and we show that  the following conditions hold:
\begin{equation}\label{eq:systemxxxx} 
\begin{cases}
\sigma(k)k\sigma^{t+1}(b)+\sigma(k^{-1})\sigma^2(k)\sigma(b)&=0,\\
\sigma^{t-1}(b)+\sigma^{-1}(k^{-2})k\sigma^{-2}(k^{-1})\sigma^{-1}(b)&=0,\\
\sigma(k)k \sigma^{t+1}(b)+ \sigma^{-1}(k^{-1})k\sigma^{-1}(b)+\sigma(k^{-1})k^{-1}\sigma(b)+\sigma^{-1}(k)k^{-1}\sigma^{t-1}(b)&=0,\\
c=\sigma(b)+\sigma^{t-1}(b)+\sigma^2(k)k^2\sigma^{t+1}(b)+\sigma^{-1}(k^{-2})k^2 \sigma^{-1}(b)&\ne 0.
\end{cases}
\end{equation}
The first two equations give the same condition, which corresponds to
\begin{equation} \label{eq:s=-s'bqs+1}
\sigma^{t+1}(b)=-k^{-1}\sigma(k^{-2})\sigma^2(k)\sigma(b).
\end{equation}
By applying $\sigma^{t-2}$ to \eqref{eq:s=-s'bqs+1}, we have
\begin{equation} \label{eq:s=-s'bqs+12}
\sigma^{-1}(b)=\sigma^{-2}(k)\sigma^{-1}(k)k^{-1}\sigma^{t-1}(b).
\end{equation}

By replacing \eqref{eq:s=-s'bqs+1} and \eqref{eq:s=-s'bqs+12} in the third equation of  \eqref{eq:systemxxxx} and after some manipulations,  we obtain
\begin{equation}\label{eq:s=-s'intermediate}
( \sigma(k^{-1})k^{-1} -\sigma(k^{-1})\sigma^2(k))\sigma(b)
 +( \sigma^{-1}(k)k^{-1}-\sigma^{-2}(k)\sigma^{-1}(k) )\sigma^{t-1}(b)=0.
\end{equation}
Applying $\sigma^2$ to \eqref{eq:s=-s'intermediate}, we derive
\begin{align}\label{eq:s=-s'condsub}
( \sigma^2(k^{-1})\sigma^3(k^{-1}) -\sigma^3(k^{-1})\sigma^4(k))\sigma^3(b) 
 +( \sigma(k)\sigma^2(k^{-1})-k\sigma(k) )\sigma^{t+1}(b)&=0,
\end{align}
and using \eqref{eq:s=-s'bqs+1}, we obtain 
\[
( \sigma^2(k^{-1})\sigma^3(k^{-1}) -\sigma^4(k^{-1})\sigma^3(k))\sigma^3(b) 
 -(-\sigma(k^{-1})\sigma^{2}(k)
+ k^{-1}\sigma(k^{-1})
 )\sigma(b)=0,
\]
which is satisfied because of the choice of $b$.

\noindent Now we show that $c$ cannot be zero.
By contradiction, let us assume that
\[
\sigma(b)+\sigma^{t-1}(b)+\sigma^{t+1}(b)+\sigma^{-1}(k^{-2})k^2 \sigma^{-1}(b)=0.
\]
Since $b$ satisfies \eqref{eq:s=-s'bqs+1} and \eqref{eq:s=-s'bqs+12}, we get 
\begin{equation} \label{eq:s=-s'condfinal}
\sigma^{3}(b)(1-\sigma^{2}(k)\sigma^{4}(k)) =- \sigma^{t+1}(b)(1-\sigma^{2}(k)k)=0,
\end{equation}
and by \eqref{eq:s=-s'condsub}
\[
(\sigma^3(k^{-1})\sigma^4(k)
-\sigma^3(k^{-1})\sigma^2(k^{-1}))
\sigma^3(b)=
 (-\sigma^{2}(k)k
+1)\sigma^2(k^{-1})\sigma(k)\sigma^{t+1}(b)
\]
Now, substituting the above equation in \eqref{eq:s=-s'condfinal} we obtain 
\[
\sigma^2(k^{-1})\sigma(k)\sigma^{3}(b)(1-\sigma^{2}(k)\sigma^{4}(k)) =- (-1+\sigma^4(k)\sigma^2(k)
)\sigma^3(k^{-1})\sigma^2(k^{-1})
\sigma^3(b),
\]
that is
\[
\sigma^3(b)(-1+\sigma^4(k)\sigma^2(k)
)(\sigma^2(k^{-1})\sigma(k)-\sigma^3(k^{-1})\sigma^2(k^{-1}))=0.
\]
However, this yields to a contradiction, since $\sigma^2(k)k \neq 1$ because of Lemma \ref{lem:auxiliary}.

\noindent At this point, we use the fact that we are dealing with an equivalence relation. Thus,  $\C_{h,t,\sigma}$ is equivalent to $\C_{k,t,\sigma^{-1}}$ if and only if $\C_{h,t,\sigma}$ is equivalent to $\C_{k^{-1},t,\sigma}$.  The assertion follows by Case I, by substituting $k$ with $k^{-1}$. 
\medskip

\noindent \underline{\textbf{Case III:} $s=t-1$}.
The first step is to prove that $\C_{k,t,\sigma^{t-1}}$ is equivalent to $\C_{k,t,\sigma}$, and then we use the case $s=1$.
We divide the discussion in two cases.
Assume that $t\equiv 2 \pmod{4}$.
We prove the existence of $a,d \in \fqn^*$ such that
\[d \psi_{k,t,\sigma^{t-1}}(x)=\psi_{k,t,\sigma}(ax),\]
that is
\begin{equation} \label{eq:casesb=0s=(t-1)s'}
\begin{cases}
d=\sigma(a)=\sigma^{t-1}(a),\\
dk\sigma(k)=-k\sigma(k)\sigma^{t+1}(a),\\
dk\sigma^{-1}(k^{-1})=-k\sigma^{-1}(k^{-1})\sigma^{-1}(a).
\end{cases}
\end{equation}
A solution to \eqref{eq:casesb=0s=(t-1)s'} is given by $a \in \F_{q^4}$ such that $a^{q^2-1}=-1$ and $d=\sigma(a)$.
As for Case I we have that denoting $\tau(x)=ax$, then
\[ \C_{k,t,\sigma}\circ \tau=\langle ax,d\psi_{k,t,\sigma^{t-1}}(x)\rangle_{\fqn}=\C_{k,t,\sigma^{t-1}}, \]
and the two codes $\C_{k,t,\sigma^{t-1}}$ and $\C_{h,t,\sigma}$ are equivalent. The assertion then follows by the Case I.
Suppose that $t\equiv 0 \pmod{4}$.
We find $b,c\in \fqn^*$ such that 
\begin{equation}\label{eq:s=t-1map}
cx=\psi_{k,t,\sigma^{t-1}}(b\psi_{k,t,\sigma}(x)),
\end{equation}
which is enough to prove that $\C_{k,t,\sigma^{t-1}}$ is equivalent to $\C_{k,t,\sigma}$.
Equation \eqref{eq:s=t-1map} reads as
\begin{align*}
cx+d\psi_{k,t,\sigma^{t-1}}(x)& =\psi_{k,t,\sigma}(ax) \\
& + (\sigma(b)-\sigma^2(k^{-1})k\sigma^{t+1}(b))x^{\sigma^{2}}\\
& +(\sigma^{t-1}(k)\sigma^{t-2}(k^{-1})\sigma^{t-1}(b)-k\sigma^{-1}(k^{-1})\sigma^{-1}(b))x^{\sigma^{t-2}} \\
& +(\sigma(b)+\sigma^{t-1}(b)-k^2\sigma(k)\sigma^{t+1}(b)-k^2\sigma^{-1}(b))x^{\sigma^{t}} \\
& +(\sigma^{t-1}(b)+\sigma^{t-2}(k)k\sigma^{-1}(b))x^{\sigma^{2t-2}}\\
& -(-\sigma(k)\sigma^2(k)\sigma(b)+k\sigma(k)\sigma^{t+1}(b))x^{\sigma^{t+2}} \\
& +(-\sigma(b)k^{-1}\sigma(k)+k^{-1}\sigma^{t-1}(k)\sigma^{t-1}(b)+k\sigma(k)\sigma^{t+1}(b)+k\sigma^{-1}(k^{-1})\sigma^{-1}(b)) x.
\end{align*}

By choosing 
$$\displaystyle b=\frac{\sigma^{-1}(\omega)}{\sigma^{-1}(h^{-1})-\sigma(h)},$$ with $\omega \in \F_{q^4}$ such that $\omega^{q^2-1}=-1$ and 
\[ c=-\sigma(b)k^{-1}\sigma(k)-k^{-1}\sigma^{-1}(k^{-1})\sigma^{t-1}(b)+k\sigma(k)\sigma^{t+1}(b)+k\sigma^{-1}(k^{-1})\sigma^{-1}(b) \ne0 \]
Equation \eqref{eq:s=t-1map} is satisfied, and so $\C_{k,t,\sigma^{t-1}}$ and $\C_{k,t,\sigma}$ are equivalent. The assertion follows by Case I.

\medskip

\noindent \underline{\textbf{Case IV:} $s=t+1$}.
Let us prove that $\C_{k,t,\sigma^{t+1}}$ is equivalent to $\C_{k^{-1},t,\sigma}$, and then we use again the case $s=1$.
We divide the discussion in two cases.
First assume that $t \equiv 0 \pmod{4}$.
We prove the existence of $a,d \in \fqn^*$ such that
\[d \psi_{k,t,\sigma^{t+1}}(x)=\psi_{k,t,\sigma}(ax),\]
that is
\begin{equation} \label{eq:casesb=0s=(t+1)s'}
\begin{cases}
-dk^{-1}\sigma(k)=\sigma(a) \\
dk^{-1}\sigma^{t-1}(k)=\sigma^{t-1}(a) \\
d=\sigma(k)k\sigma^{t+1}(a) \\
d=\sigma^{-1}(k^{-1})\sigma^{-1}(a) .
\end{cases}
\end{equation}
A solution to \eqref{eq:casesb=0s=(t+1)s'} is given by $a=l k$, with $l\in \F_{q^4}$ such that $l^{q^2+1}=-1$, and $d=-k\sigma(l)$. 
As for Case I we have that denoting $\tau(x)=ax$, then
\[ \C_{k,t,\sigma}\circ \tau=\langle ax,d\psi_{k,t,\sigma^{t+1}}(x)\rangle_{\fqn}=\C_{k,t,\sigma^{t-1}}, \]
and the two codes $\C_{k,t,\sigma^{t+1}}$ and $\C_{h,t,\sigma}$ are equivalent. The assertion then follows by the Case I.
Finally, assume $t \equiv 2 \pmod{4}$.
We find $b,c \in \fqn^*$ such that 
\begin{equation}\label{eq:s=t+1prima} cx=\psi_{k,t,\sigma^{t+1}}(b\psi_{k^{-1},t,\sigma}(x)), \end{equation}
that is
\begin{align*}
cx&= (\sigma(k)k\sigma^{t+1}(b)-\sigma(k^{-1})\sigma^2(k)\sigma(b))x^{\sigma^{2}}\\
&+(\sigma^{t-1}(b)+\sigma^{-1}(k^{-2})k\sigma^{t-2}(k)\sigma^{-1}(b))x^{\sigma^{t-2}}\\
&+(\sigma(k)k \sigma^{t+1}(b)+ \sigma^{-1}(k^{-1})k\sigma^{-1}(b)-\sigma(k^{-1})k^{-1}\sigma(b)+\sigma^{t-1}(k^{-1})k^{-1}\sigma^{t-1}(b))x^{\sigma^{t}}\\
&+(\sigma^{-1}(k^{-1})k\sigma^{-1}(b)+\sigma^{t-1}(k^{-1})\sigma^{-2}(k)\sigma^{t-1}(b))x^{\sigma^{2t-2}}\\
&+(\sigma(b)+\sigma^{t-1}(b)-\sigma^2(k)k^2\sigma^{t+1}(b)-\sigma^{-1}(k^{-2})k^2 \sigma^{-1}(b)) x\\
&+(\sigma(b)-k\sigma^2(k)\sigma^2(k^{-1})\sigma^{t+1}(b))x^{\sigma^{t+2}}.
\end{align*}
Arguing as in Case II, 
we can choose 
\[ b=\frac{\sigma^{-1}(\omega)}{h\sigma(h^{-1})-h\sigma^{-1}(h)}, \]
with $\omega\in \F_{q^4}$ such that $\omega^{q^2-1}=-1$, and  $$c=\sigma(b)+\sigma^{t-1}(b)-\sigma^2(k)k^2\sigma^{t+1}(b)-\sigma^{-1}(k^{-2})k^2 \sigma^{-1}(b)\ne0.$$ With this choice of $b$ and $c$,  \eqref{eq:s=t+1prima} is satisfied, so that the two codes $\cC_{k,t,\sigma}$ and $\cC_{k^{-1},t,\sigma^{t+1}}$ are equivalent. The assertion follows again by Case I.
\end{proof}

Theorem \ref{th:equivspecialcase} characterizes all the cases in which two codes $\cC_{h,t,\sigma}$ and $\cC_{k,t,\theta}$ are equivalent, when we restrict to have $\theta=\sigma^{s}$, for $s \in \{1,t-1,t+1,2t-1\}$. 

We now analyze the  cases in which  $\theta=\sigma^{s}$ with $s \notin \{1,t-1,t+1,2t-1\}$, showing that two codes of the form $\cC_{h,t,\sigma}$ and $\cC_{k,t,\theta}$ are never equivalent. The following computations will be crucial  for the study of these remaining cases. These computations aim to describe the codes obtained from $\cC_{h,t,\sigma}$ together with some automorphism applied to it. The first thing to observe is that  $\cD_{h,t,\sigma}^{(j)}:=\cC_{h,t,\sigma}+\sigma^{j}(\cC_{h,t,\sigma})$ can have dimension $2$, $3$ or $4$, and when $j$ is coprime with $n=2t$, then the only possibility is that $\dim_{\fqn}(\cD)=4$; cf. Proposition \ref{prop:not_Gabidulin}. However, an interesting object comes when we choose $j=t$. Indeed, one can easily check that 
\begin{equation}\label{eq:c_plus_ct}\cD_{h,t,\sigma}^{(t)}=\langle x, x^{\sigma^t}, x^{\sigma}+h\sigma(h)x^{\sigma^{t+1}}, x^{\sigma^{t-1}}+h\sigma^{-1}(h^{-1})x^{\sigma^{2t-1}} \rangle_{\fqn}. \end{equation}
This object has a nice and compact description and it will be crucial for determining the inequivalence of codes of  $\cC_{h,t,\sigma}$ and $\cC_{h,t,\theta}$. In order to derive these results, we need to apply another automorphism to $\cD_{h,t,\sigma}^{(t)}$. 

\begin{proposition}\label{prop:d_plus_sigma}
 The following hold:
 \begin{align*}
     s_1^{\sigma^s}(\cD_{h,t,\sigma}^{(t)})&=6, \qquad \mbox{ for } s \in \{1,-1,t-1,t+1\},\\
        s_1^{\sigma^s}(\cD_{h,t,\sigma}^{(t)})&\geq 7, \qquad \mbox{ for } s \not\in \{1,-1,t-1,t+1\},
 \end{align*}
\end{proposition}

\begin{proof} In order to lighten the notation, let us write $\cD:=\cD_{h,t,\sigma}^{(t)}$. With straightforward computations we derive
  \begin{align*}\cD+\sigma(\cD)&=\langle x, x^\sigma,x^{\sigma^t}, x^{\sigma^{t+1}}, x^{\sigma^{2}}+\sigma(h)\sigma^{2}(h)x^{\sigma^{t+2}}, x^{\sigma^{t-1}}+h\sigma^{-1}(h^{-1})x^{\sigma^{2t-1}}\rangle_{\fqn}, \\
    \cD+\sigma^{t-1}(\cD)&=\langle x, x^{\sigma^{t-1}},x^{\sigma^t},x^{\sigma^{2t-1}}, x^{\sigma}+h\sigma(h)x^{\sigma^{t+1}}, x^{\sigma^{2t-2}}+\sigma^{-1}(h^{-1})\sigma^{-2}(h)x^{\sigma^{t-2}}  \rangle_{\fqn},\\
    \cD+\sigma^{t+1}(\cD)&=\langle x, x^{\sigma}, x^{\sigma^t},  x^{\sigma^{t+1}},  x^{\sigma^{t+2}}+\sigma(h^{-1})\sigma^2(h^{-1})x^{\sigma^{2}}, x^{\sigma^{t-1}}+h\sigma^{-1}(h^{-1})x^{\sigma^{2t-1}} \rangle_{\fqn},\\
      \cD+\sigma^{2t-1}(\cD)&=\langle x, x^{\sigma^{t-1}}, x^{\sigma^t}, x^{\sigma^{2t-1}}, x^{\sigma}+h\sigma(h)x^{\sigma^{t+1}}, x^{\sigma^{t-2}}+\sigma^{-1}(h)\sigma^{-2}(h^{-1})x^{\sigma^{2t-2}} \rangle_{\fqn}.
     \end{align*}
     In all the four cases, the generators of the code have disjoint $\sigma$-supports, showing that the dimension is $6$.
     
     Let us assume now that $s\notin\{1.t-1,t+1,2t-1\}$. Then
        \begin{align*}\cD+\sigma^{s}(\cD)=\langle x, x^{\sigma^t}, x^{\sigma^s},x^{\sigma^{s+t}},f_1(x), f_2(x), f_3(x), f_4(x) \rangle_{\fqn}, \end{align*}
        where 
        $$\begin{array}{rlrl}
            f_1(x)\!\!\!\!&= x^{\sigma}+h\sigma(h)x^{\sigma^{t+1}}, & f_2(x)\!\!\!\!&= x^{\sigma^{t-1}}+h\sigma^{-1}(h^{-1})x^{\sigma^{2t-1}},\\
            f_3(x)\!\!\!\!&=x^{\sigma^{s+1}}+\sigma^s(h)\sigma^{s+1}(h)x^{\sigma^{t+s+1}}, \qquad & f_4(x)\!\!\!\! &=x^{\sigma^{s+t+1}}+\sigma^s(h)\sigma^{s-1}(h^{-1})x^{\sigma^{s-1}}.
        \end{array}$$
    First, observe that the $\sigma$-supports of the first six polynomials are pairwise distinct, thus they are $\fqn$-linearly independent and  $\cV:=\langle x, x^{\sigma^t}, x^{\sigma^s},x^{\sigma^{s+t}},f_1(x), f_2(x)\rangle_{\fqn}$ is  a $6$-dimensional $\fqn$-subspace. Moreover, the union of their $\sigma$-supports is 
    $$\mathcal P=\{0,s,t,s+t,1,t+1,t-1,2t-1\}.$$
    Now, assume by contradiction that $f_3(x),f_4(x)\in\cV$.
    This implies that $\supp_\sigma(f_3(x))=\{s+1,t+s+1\}\subseteq\mathcal P$.
    Due to the assumptions on $s$, the only case for which this is possible is 
    $s\in\{t-2,2t-2\}$. In both cases, it is immediate to see that $\supp_\sigma(f_4(x))=\{t-3,2t-3\}$ is not contained in $\mathcal P=\{0,t-2,t,2t-2,1,t+1,t-1,2t-1\}$, yielding a contradiction.  
\end{proof}

\begin{remark}
 From the computations of $\cD_{h,t,\sigma}^{(t)}+\sigma^{j}(\cD_{h,t,\sigma}^{(t)})$, for $j \in \{1,t-1,t+1,2t-1\}$ in the proof of  Proposition \ref{prop:d_plus_sigma}, we can immediately observe that 
 \begin{align*}
     \cD_{h,t,\sigma}^{(t)}+\sigma(\cD_{h,t,\sigma}^{(t)})&=\cD_{h,t,\sigma}^{(t)}+\sigma^{t+1}(\cD_{h,t,\sigma}^{(t)}) \\
     \cD_{h,t,\sigma}^{(t)}+\sigma^{t-1}(\cD_{h,t,\sigma}^{(t)}) &=\cD_{h,t,\sigma}^{(t)}+\sigma^{2t-1}(\cD_{h,t,\sigma}^{(t)}).
 \end{align*}
 This also suggests to consider
 another space that is highly nongeneric, which is \begin{align*}
     \cT:= & \cD_{h,t,\sigma}^{(t)}+\sigma(\cD_{h,t,\sigma}^{(t)})+\sigma^{t-1}(\cD_{h,t,\sigma}^{(t)})+\sigma^{t+1}(\cD_{h,t,\sigma}^{(t)})+\sigma^{2t-1}(\cD_{h,t,\sigma}^{(t)}) \\
     =& \cD_{h,t,\sigma}^{(t)}+\sigma(\cD_{h,t,\sigma}^{(t)})+\sigma^{t-1}(\cD_{h,t,\sigma}^{(t)}).
 \end{align*} A straightforward computation shows that 
 $$ \cT= \langle  x, x^\sigma,  x^{\sigma^{t-1}},x^{\sigma^t}, x^{\sigma^{t+1}}, x^{\sigma^{2t-1}}, x^{\sigma^{2}}+\sigma(h)\sigma^{2}(h)x^{\sigma^{t+2}},x^{\sigma^{t-2}}+\sigma^{-1}(h)\sigma^{-2}(h^{-1})x^{\sigma^{2t-2}} \rangle_{\fqn}.$$
 Since all the generators have disjoint $\sigma$-supports, then $\dim_{\fqn} (\cT)=8$. Moreover, we can express this code as
 $$ \cT= \cC_{h,t,\sigma}+\sigma(\cC_{h,t,\sigma})+\sigma^{t-1}(\cC_{h,t,\sigma})+\sigma^t(\cC_{h,t,\sigma})+\sigma^{t+1}(\cC_{h,t,\sigma})+\sigma^{2t-1}(\cC_{h,t,\sigma}).$$
\end{remark}

\begin{corollary}\label{cor:inequivalence_newMRD}
 If $s \notin \{1,-1,t-1,t+1\}$, then $\cC_{h,t,\sigma}$ and $\cC_{k,t,\sigma^s}$ are not equivalent.
\end{corollary}

\begin{proof}
 Assume by contradiction that $\cC_{h,t,\sigma}$ and $\cC_{k,t,\sigma^s}$ are  equivalent. Hence, by Proposition \ref{prop:equivalence_sum} also 
 $\cC_{h,t,\sigma}+\sigma^{st}(\cC_{h,t,\sigma})$ and $\cC_{k,t,\sigma^s}+\sigma^{st}(\cC_{k,t,\sigma^s})=\cD_{k,t,\sigma^s}^{(t)}$ are equivalent. Moreover, observe that $\sigma^{st}=\sigma^t$, which means that $\cC_{h,t,\sigma}+\sigma^{st}(\cC_{h,t,\sigma})=\cD_{h,t,\sigma}^{(t)}$. Using again Proposition \ref{prop:equivalence_sum}, we also derive that 
$s_1^{\sigma^s}(\cD_{h,t,\sigma}^{(t)})= s_1^{\sigma^s}(\cD_{h,t,\sigma^s}^{(t)})$, but this contradicts Proposition \ref{prop:d_plus_sigma}.
\end{proof}

As a byproduct of all the results obtained in this section, we can also determine the exact number and a lower bound on the number of equivalence classes of the codes $\C_{h,t,\sigma}$. 

\begin{theorem}\label{th:numexact}
The number of  equivalence classes of the codes $\cC_{h,t,\sigma}$ is 
$$ \frac{\varphi(t)}{4rtj_t}\sum_{i=0}^{2rt-1} \deg ( \gcd(x^{p^{rt}+1}+1,x^{j_t(p^i-1)}-1)), $$
where $$ j_t=\begin{cases}
2 & \mbox{ if } t \not\equiv 2 \, \pmod{4},\\
p^{2r}+1 & \mbox{ if } t \equiv 2 \, \pmod{4}, \\
\end{cases}$$ and $\varphi$ is the Euler's totient function.
\end{theorem}
\begin{proof}
We start by determining the number of orbits of $A_{h,t}$ under the action of $\Aut(\fqn)$, where
$$A_{h,t}:=\begin{cases}
\{h,-h\} & \mbox{ if } t \not\equiv 2 \, \pmod{4}\\
\{lh \, : \, l^{q^2+1}=1 \} & \mbox{ if } t \equiv 2 \, \pmod{4}.
\end{cases}$$
First, observe that we can reduce ourselves to study this action. For each $\rho \in \Aut(\fqn)$, either we have $\rho(A_{h,t})=A_{h,t}$ or $\rho(A_{h,t})\cap A_{h,t}=\emptyset$. This is because in both cases $A_{h,t}$ is of the form $K h$, where $K$ is the multiplicative subgroup of $\fqn^*$ given by
$$ K=\begin{cases}
\{\pm 1\} & \mbox{ if } t \not\equiv 2 \, \pmod{4}, \\
\{l \, : \, l^{q^2+1}=1\} & \mbox{ if } t \equiv 2 \, \pmod{4}.
\end{cases}$$
Therefore, we have that $\rho(K)$ is a subgroup of $\fqn^*$ with the same order. Since $\fqn^*$ is cyclic, then $\rho(K)=K$ and thus $\rho(Kh)=K\rho(h)$. Hence, we can just consider whether $\rho(h)\in A_{h,t}$ or not. In particular, $G:=\Aut(\fqn)$ acts on the set
$$ X:=\{A_{h,t}\, :\, \sigma^t(h)h=-1  \},$$
and denote by $|X/G|$ the number of orbits of the action of $G$ on $X$. This can be computed  by means of Burnside's lemma, which reads as
\begin{equation}\label{eq:burnside}|X/G|=\frac{1}{|G|}\sum_{g \in G} |X^g|,\end{equation}
where $X^g=\{x \in X \, : \, g(x)=x\}$. 

Let $x:=A_{h,t}$ and let $\theta$ denote the $p$-Frobenius automorphism of $\fqn$, that is $\theta(\alpha)=\alpha^p$ for every $\alpha \in \fqn$. Then, each $\rho$ can be written as $\theta^i$ for some $i \in \ZZ{2rt}$. In particular,
\begin{align*}X^{\theta^i}&=\left\{A_{h,t} \,: \, \theta^{rt}(h)=-h^{-1}, \theta^{i}(h) \in A_{h,t} \right\} \\
&=\bigcup_{\kappa  \in K}\left\{A_{h,t} \,: \, \theta^{rt}(h)=-h^{-1}, \theta^{i}(h)=\kappa h \right\} \\
&= \bigcup_{\kappa  \in K}\left\{A_{h,t} \,: \, h^{p^{rt}+1}+1=0, h^{p^i}-\kappa h=0 \right\}\\
&= \Big\{A_{h,t} \,: \, h^{p^{rt}+1}+1=0, \prod_{\kappa \in K}(h^{p^i-1}-\kappa )=0 \Big\}\\
&= \left\{A_{h,t} \,: \, h^{p^{rt}+1}+1=0, h^{j_t(p^i-1)}-1=0 \right\}.\end{align*}
Since we have to consider only one representative for each set $A_{h,t}$, we deduce
\begin{equation}\label{eq:gcd_degree} |X^{\theta^i}|=\frac{1}{j_t}\deg(\gcd(x^{p^{rt}+1}+1,x^{j_t(p^i-1)}-1)).
\end{equation}
Thus, combining this with the fact that 
$$\varphi(2t)=\begin{cases}
2\varphi(t) & \mbox{ if } t \mbox{ is even }\\
\varphi(t) & \mbox{ if } t \mbox{ is odd} 
\end{cases}$$ 
we can conclude that the number of inequivalent codes of the form $\C_{h,t,\sigma}$ in $\cL_{n,q}$ is
$$ \frac{\varphi(t)}{4rtj_t}\sum_{i=0}^{2rt-1} \deg ( \gcd(x^{p^{rt}+1}+1,x^{j_t(p^i-1)}-1)).$$
\end{proof}

\begin{remark}
In some cases we can state the number of inequivalent classes of Theorem \ref{th:numexact} in a more direct way.
Assume that either $p\equiv 1 \,\pmod{4}$ or $rt$ is even. Then we have
$$\begin{cases}
p^{rt}+1 &\equiv 2 \, \pmod{4}, \\
j_t(p^{i}-1) \!\!\!&\equiv 0 \, \pmod{4},
\end{cases}$$
and by \cite[Remark 4.1]{mcguiretrinomials} we obtain that 
$$|X^{\theta^i}|=\gcd(p^{rt}+1,j_t(p^i-1)).$$
This is for instance the case for $t\equiv 2 \, \pmod{4}$, that is $t=2t'$ with $t'$ odd, in which $j_t=p^{2r}+1$. It is easy to see that $p^{2r}+1$ divides $p^{2rt'}+1$, and we have
$$ |X^{\theta^i}|=(p^{2r}+1)\gcd(p^{2r(t'-1)}-p^{2r(t'-2)}+\ldots +1,p^i-1).$$
So, the number of inequevalent classes of the codes $\C_{h,t,\sigma}$ is
$$ \frac{\varphi(t)}{4rt}\sum_{i=0}^{2rt-1} \gcd(p^{2r(t'-1)}-p^{2r(t'-2)}+\ldots +1,p^i-1). $$
\end{remark}

Unfortunately, the result of Theorem \ref{th:numexact} is quite implicit. However, we can give the following more explicit lower bound.

\begin{theorem}\label{th:lowbound}
The number of  equivalence classes of the codes $\cC_{h,t,\sigma}$ is at least
$$ \begin{dcases}
\left\lfloor\dfrac{\varphi(t)(q^t+1)}{4rt(q^2+1)}\right\rfloor & \mbox{ if } t\equiv 2 \pmod{4}, \\
\left\lfloor\dfrac{\varphi(t)(q^t+1)}{8rt}\right\rfloor & \mbox{ if }  t \not\equiv 2 \pmod{4}, 
\end{dcases} $$
where $\varphi$ is the Euler's totient function.
\end{theorem}
\begin{proof}
 The number of pairs $(h,\sigma)\in \F_{q^n}\times \Gal(\fqn/\fq)$ such that the codes $\cC_{h,t,\sigma}$ are proved to be MRD is $(q^t+1)\varphi(2t)$. Moreover, combining Theorem \ref{th:equivspecialcase} with Corollary \ref{cor:inequivalence_newMRD}, we know that each equivalence class has $N$ elements, where
 $$N\leq \begin{cases}
16rt & \mbox{ if } t\equiv 0  \pmod{4}, \\
8rt & \mbox{ if }  t \equiv 1,3 \pmod{4}, \\
8rt(q^2+1) & \mbox{ if } t\equiv 2 \pmod{4}. 
 \end{cases} $$
 This follows from the fact that in Theorem \ref{th:equivspecialcase} the exponent $s$ can be equal to $t-1$ and $t+1$ if and only if $t$ is even -- otherwise $\sigma^s$ is not a generator of $\Gal(\fqn/\fq)$. Thus, combining this with the fact that 
 $$\varphi(2t)=\begin{cases}
2\varphi(t) & \mbox{ if } t \mbox{ is even }\\
\varphi(t) & \mbox{ if } t \mbox{ is odd} 
\end{cases}$$ 
we obtain a lower bound on the number of equivalence classes of the codes $\cC_{h,t,\sigma}$:
$$ \dfrac{\varphi(2t)(q^t+1)}{N}\geq \begin{dcases}
\dfrac{\varphi(t)(q^t+1)}{4rt(q^2+1)} & \mbox{ if } t\equiv 2 \pmod{4}, \\
\dfrac{\varphi(t)(q^t+1)}{8rt} & \mbox{ if }  t \not\equiv 2 \pmod{4}. 
\end{dcases} $$
\end{proof}

\subsection{Right idealizer and adjoint}
In this section we complete the study on the new family of codes, by analyzing their right idealizer and their adjoint.

\begin{proposition}\label{prop:right_idealizer}
Let $t\geq 5$. The right idealizer of $\mathcal{C}_{h,t,\sigma}$ is
\[ R(\cC_{h,t,\sigma})=\begin{cases} \left\{ax \colon a \in \F_{q^2}\right\}
& \mbox{ if $t$ is even, } \\ 
\big\{ax+b\psi_{h,t,\sigma}(x) \colon a \in \F_{q}, b=\frac{\delta}{\sigma(h)-\sigma^{-1}(h^{-1})},\delta^q+\delta=0\big\} & \mbox{ if $t$ is odd. } \end{cases} \]
In particular, $R(\mathcal{C}_{h,t,\sigma})\simeq \F_{q^2}$.
\end{proposition}

\begin{proof}
Let $f(x)$  be an element in $R(\mathcal{C}_{h,t,\sigma})$.
Hence, $f(x)\in R(\mathcal{C}_{h,t,\sigma})$, so that there exists $a,b \in \F_{q^{2t}}$ such that
\[ f(x)=ax+b\psi_{h,t,\sigma}(x). \]
Moreover, $\psi_{h,t,\sigma}(f(x))=\psi_{h,t,\sigma}(ax)+\psi_{h,t,\sigma}(b\psi_{h,t,\sigma}(x)) \in \C_{h,t,\sigma}$, that is there exist $c,d \in \F_{q^{2t}}$ such that
\[ \psi_{h,t,\sigma}(ax)+\psi_{h,t,\sigma}(b\psi_{h,t,\sigma}(x))=cx+d\psi_{h,t,\sigma}(x). \]
So that, one can argue as in the second part of the proof of \cite[Theorem 4.2]{longobardi2021large}, obtaining that $b=0$ if $t$ is even and 
\[ b=\frac{\delta}{\sigma(h)-\sigma^{-1}(h^{-1})}, \]
otherwise.
In the first case, we have again System \eqref{eq:systems=1} with $h=k$, which implies $a \in \F_{q^2}$.
In the latter case, System \eqref{eq:systems=1} with $h=k$ implies $a \in \F_q$.
It is easy to check that these maps are in $R(\C_{h,t,\sigma})$.
Since $R(\C_{h,t,\sigma})$ is a finite field (see \cite[Corollary 5.6]{lunardon2018nuclei}), and we have   $|R(\mathcal{C}_{h,t,\sigma})|=q^2$, then $R(\mathcal{C}_{h,t,\sigma})\simeq \F_{q^2}$.
\end{proof}

\begin{remark}
As a consequence of Proposition \ref{prop:right_idealizer},  we may deduce an alternative proof of Proposition \ref{prop:not_Gabidulin}. Indeed, the right idealizer of a Gabidulin code is isomorphic to $\fqn$ (see \cite{morrison2014equivalence,liebhold2016automorphism}) and it is well-known that  equivalent codes have isomorphic idealizers; see \cite[Proposition 4.1]{lunardon2018nuclei}.
 However, the same approach does not work for showing that the codes $\cC_{h,t,\sigma}$ are not equivalent to $2$-dimensional twisted Gabidulin codes, since in this case the right idealizers are isomorphic. 
\end{remark}

\begin{proposition}
 If $t\geq 5$, then the code $\C_{h,t,\sigma}^\top$ is equivalent to $\C_{h,t,\sigma}$.
\end{proposition}
\begin{proof}
Note that 
\[ \C_{h,t,\sigma}^\top=\langle x, {\psi}_{h,t,\sigma}^\top(x) \rangle_{\fqn}, \]
which turns out to be equivalent to $\C'=\langle x, g(x) \rangle_{\fqn}$, with \[g(x)=h{\psi}_{h,t,\sigma}^\top(h^{-1}x)=x^\sigma-x^{\sigma^{t-1}}-h\sigma(h)x^{\sigma^{t+1}}+h\sigma^{-1}(h^{-1})x^{\sigma^{2t-1}}.\]
As already done in Theorem \ref{th:equivspecialcase}, it is enough to find $b,c \in \fqn^*$ such that
\begin{equation}\label{eq:condbcadjoint} cx=\psi_{h,t,\sigma}(bg(x)). \end{equation}
Choosing
\[ b=\frac{\sigma^{-1}(h)}{\sigma(h)\sigma^{-1}(h)-1}, \]
and
\[ c=\sigma(b)h^{-1}\sigma(h)-\sigma^{t-1}(b)h^{-1}\sigma^{-1}(h^{-1})-\sigma^{t+1}(b)h\sigma(h)+\sigma^{-1}(b)h\sigma^{-1}(h^{-1}), \]
and performing similar calculations as in the proof of \cite[Theorem 4.6]{longobardi2021large}, \eqref{eq:condbcadjoint} turns out to be verified and so $\C_{h,t,\sigma}^\top$ and $\C_{h,t,\sigma}$ are equivalent.
\end{proof}

\section{Conclusions and open problems}\label{sec:conclusions}
In this paper we have provided a large family of rank-metric codes, which contains properly the codes found in  \cite{longobardizanellascatt} and in \cite{longobardi2021large}. These codes are $\fqn$-linear of dimension $2$ in the space $\cL_{n,q}$, where $n=2t$ and $t$ is any integer greater than $2$, and we proved that they are MRD. 
We have also studied exhaustively the equivalence of such codes for $t\ge 5$, characterizing their equivalence classes.
These codes turn out to be inequivalent to any other construction known so far, and hence they are really new.

However, there are still some open problems which are  related to the results obtained in this paper.

\begin{itemize}
    \item When $t=3$ and $\sigma\colon x \in \F_{q^6} \longmapsto x^q \in \F_{q^6}$, the problem of  code equivalence of $\C_{h,3,\sigma}$  with the other known $\F_{q^6}$-linear MRD codes in $\cL_{6,q}$ was investigated in \cite{BZZ}.
    It is easy to see that in this case $\C_{h,3,\sigma^{-1}}=\C_{h^{-1},3,\sigma}$, so that  the results in \cite{BZZ} complete the study of the equivalence between $\C_{h,3,\sigma}$ and the other known $\F_{q^6}$-linear MRD codes in $\cL_{6,q}$.
    However, it would be interesting to get the number of inequivalent MRD codes of the form $\C_{h,3,\sigma}$.
    \item The equivalence study for the case $t=4$ is open. As already mentioned before, more complicated calculations should be performed to analyze the equivalence of the code $\C_{h,4,\sigma}$ with the other known $\F_{q^8}$-linear MRD codes in $\cL_{8,q}$ and with the other codes of shape $\C_{k,4,\sigma'}$.
    \item In \cite[Section 3]{BZZ} (see also \cite{ZZ}), it was also proved that the $\fq$-linear set defined by $\psi_{h,3,\sigma}(x)$
    \[ L_{\psi_{h,3,\sigma}}=\{ \langle (x,\psi_{h,3,\sigma}(x)) \rangle_{\F_{q^6}} \colon x \in \F_{q^6}^* \} \]
    is not $\mathrm{P}\Gamma\mathrm{L}(2,q^6)$-equivalent to any known scattered $\F_q$-linear set in $\mathrm{PG}(1,q^6)$, except for the case $h \in \F_q$ and $q$ is a power of $5$.
    Whereas, in \cite[Section 5]{longobardi2021large} the authors proved that the polynomials $\psi_{h,t,\sigma}(x)$, with  $\sigma\colon x \in \F_{q^n} \mapsto x^q \in \F_{q^n}$, define a large class of scattered $\fq$-linear sets, so that there must be new examples of scattered linear sets defined by polynomials of the form $\psi_{h,t,\sigma}(x)$. We think that could be of interest to generalize the above results of \cite{longobardi2021large} by replacing $\sigma$ by any generator of $\mathrm{Gal}(\fqn/\fq)$ and also to investigate such equivalence issue also to the case $t=4$, which is the first open case.
    \item We have already mentioned in Section \ref{sec:rankmetric_linearized} that the same skew algebra $\cL_{n,\sigma}$ over a field $\LL$ can  be used to represent the algebra $\K^{n\times n}$, where $\LL/\K$ is a cyclic Galois extension of degree $n$ with Galois group of order $n$ generated by $\sigma$. Here, it is possible to define the same code $\cC^\LL_{h,t,\sigma}$, by picking an element $h\in\LL$ such that $\sigma^t(h)h=-1$. By using the arguments provided in \cite[Section 2]{elmaazouz2021}, one can show that the code $\cC^\LL_{h,t,\sigma}$ is MRD\footnote{In the case of infinite fields, the definition of an MRD code is slightly different: a rank-metric code $\C\in \K^{n \times n}$ that is $\mathbb E$-linear, for some subfield $\mathbb E$ of $\K$ of finite index, is  MRD if $\dim_{\mathbb E}(\C)=n[\K:\mathbb E](n-d(\cC)+1)$. Note that this definition coincides with the definition of $\F_{p^i}$-linear MRD codes in $\fq^{n \times n}$, where $\F_{p^i}$ is a subfield of $\fq$.} whenever $\LL/\K$ is an unramified extension of non-Archimedean local fields. Indeed, if $\fqn/\fq$ is the corresponding extension of the residue fields,  whenever we take an element $a\psi_{h,t,\sigma}(x)+bx\in\cC_{h,t,\sigma}^\LL$ for some $a,b \in \mathcal O_{\LL}$ where one of them has valuation $0$, the reduction modulo the maximal ideal of $\mathcal O_{\LL}$ gives a nonzero element of the MRD code relative to the extension $\fqn/\fq$. Since the rank of the  map reduced modulo the maximal ideal cannot increase, then also the minimum distance of $\cC_{h,t,\sigma}$ is at most the minimum distance of $\cC^\LL_{h,t,\sigma}$, showing that $\cC^\LL_{h,t,\sigma}$ must be MRD. It seems reasonable to think that the assumption on $\LL/\K$ to be an unramified extension of non-Archimedean local fields can be removed, and it would be interesting to prove that $\cC^\LL_{h,t,\sigma}$ is always MRD. This could  also provide an alternative proof of the fact that the codes  $\cC_{h,t,\sigma}$ is MRD not relying on the scatteredness of the $\sigma$-polynomial $\psi_{h,t,\sigma}(x)$.
\end{itemize}

\section*{Acknowledgements}
The first author was partially supported by \emph{Swiss National Science Foundation}, through
grant no. 187711. The research of the last two authors was supported by the Italian National Group for Algebraic and Geometric Structures and their Applications (GNSAGA - INdAM). The last author was also supported by the project ``VALERE: VAnviteLli pEr la RicErca" of the University of Campania ``Luigi Vanvitelli''.

\bibliographystyle{abbrv}
\bibliography{main}

\end{document}